\documentclass[a4paper,11pt]{article}
\usepackage{amsthm,amsmath,amsfonts,amssymb}

\newcommand{\one}{\mathbf{1}}
\newcommand{\lap}{\Delta}
\newcommand{\cH}{\mathcal{H}}
\newcommand{\C}{\mathbb C}
\newcommand{\R}{\mathbb R}
\newcommand{\N}{\mathbb N}

\newcommand{\rkl}{\rangle}
\newcommand{\lkl}{\langle}
\newcommand{\nn}{\nonumber}
\newcommand{\ds}{\displaystyle}

\newcommand{\Hel}{\mathcal{H}_{el}}
\newcommand{\F}{\mathcal{F}}
\newcommand{\cL}{\mathcal{L}}
\newcommand{\cO}{\mathcal{O}}
\newcommand{\co}{o}
\newcommand{\cR}{\mathcal{R}}
\newcommand{\cS}{\mathcal{S}}
\newcommand{\hel}{H_{\mathrm{el}}}
\newcommand{\uf}{\underline{f}}
\newcommand{\uh}{\underline{h}}

\newcommand{\Ran}{\mathrm{Ran}}

\newcommand{\supp}{\mathrm{supp}}
\newcommand{\gs}{\Phi_{\alpha}}
\newcommand{\go}{\Phi_{0}}

\newcommand{\fA}{\mathbf{A}}
\newcommand{\fn}{\mathbf{0}}
\newcommand{\fE}{\mathbf{E}}
\newcommand{\fp}{\mathbf{p}}
\newcommand{\fq}{\mathbf{q}}
\newcommand{\fk}{\mathbf{k}}
\newcommand{\fx}{\mathbf{x}}
\newcommand{\fK}{\mathbf{K}}
\newcommand{\fX}{\mathbf{X}}
\newcommand{\fG}{\mathbf{G}}
\newcommand{\fF}{\mathbf{F}}
\newcommand{\Wdip}{W^{\rm dip}}
\newcommand{\feps}{\mbox{\boldmath $\varepsilon$}}
\newcommand{\sprod}[2]{\left\langle #1,#2 \right\rangle}
\newcommand{\expect}[1]{\langle #1\rangle}
\newcommand{\eps}{\varepsilon}
\newcommand{\ph}{\varphi}
\newcommand{\hap}{\cH_+^{\alpha}}
\newcommand{\pap}{P_+^{\alpha}}

\newtheorem{theorem}{Theorem}[section]
\newtheorem{lemma}[theorem]{Lemma}

\newtheorem{proposition}[theorem]{Proposition}
\theoremstyle{plain}
\begin{document}
\title{On the Atomic Photoeffect in Non-relativistic QED}

\author{\vspace{5pt} Marcel Griesemer$^1$ and Heribert Zenk$^2$\\
\vspace{-4pt}\small{$1.$ Fachbereich Mathematik, Universit\"at
Stuttgart,} \\
\small{D--70569 Stuttgart, Germany}\\
\vspace{-4pt}\small{$2.$ Mathematisches Institut,
Ludwig-Maximilians-Universit\"at
M\"unchen,} \\
\small{D-80333 M\"unchen, Germany}}
%\date{}
\maketitle
\begin{abstract}
In this paper we present a mathematical analysis of the photoelectric effect for
one-electron atoms in the framework of non-relativistic QED. We treat
photo-ionization as a scattering process where in the remote past
an atom in its ground state is targeted by one or several photons,
while in the distant future the atom is ionized and the electron 
escapes to spacial infinity. 
Our main result shows that the ionization probability, to leading
order in the fine-structure constant, $\alpha$, is correctly given by
formal time-dependent perturbation theory, and, moreover, that the
dipole approximation produces an error
of only sub-leading order in $\alpha$. In this sense, the dipole
approximation is rigorously justified.
\end{abstract}

%\noindent
%{\bf MSC:} 81Q10, 81V10, 47N50.
%\noindent
%{\bf Keywords:} Photoelectric Effect, Scattering Theory, QED,
%Dipole Approximation, Asymptotic Creation Operators

%\tableofcontents

%\newpage
\section{Introduction}

Even today, more than 100 years after its discovery by Hertz,
Hallwachs and Lenard, the phenomenon of photoionization 
is still investigated, both experimentally and
theoretically \cite{BrabecKrausz2000, Arendt2007}. This research is driven by novel experimental
techniques that allow for the production of very strong and
ultrashort laser pulses. In contrast, the photo electric effect in the
early experiments is produced by weak, non-coherent radiation of high
frequency. There is a third physical regime, where the radiation is
weak, of high frequency, and  \emph{coherent}.
This regime is the subject of the present paper. We consider one-electron atoms within the
standard model of non-relativistic QED, and we present a mathematically
rigorous analysis of the ionization process caused by the impact of
finitely many photons. Improving on earlier
results concerning more simplified models, we show that the probability
of ionization, to leading order in the fine-structure constant, is proportional to the number of photons, and, in the case
of a single photon, it is given correctly by the rules of
formal (time-dependent) perturbation theory. It turns out that the dipole approximation produces an
error of subleading order, which provides a rigorous justification of this
popular approximation.

Let's briefly recall the standard model of one-electron atoms
within non-relativistic QED. More elaborate descriptions may be
found elsewhere \cite{Gri2006, Spohn2004}. States of arbitrarily many transversal photons
are described by vectors in the symmetric Fock space
$$
         \F:=\bigoplus_{n\geq 0}
         S_n\big[\otimes^{n}L^2(\R^3\times\{1,2\})\big]
$$
over $L^2(\R^3\times\{1,2\})$. Here $S_n$ denotes the
projection of $L^2(\R^3\times\{1,2\})^n$ onto the subspace of all
symmetric functions of
$(\fk_1,\lambda_1),\ldots,(\fk_n,\lambda_n)\in\R^3\times\{1,2\}$, and
$S_{0}L^2(\R^3\times\{1,2\}):=\C$. We shall use $\Omega$ to denote the
vacuum vector $(1,0,\ldots)\in\F$. $N_f$ is the number operator
in $\F$, and $H_f={\mathrm d}\Gamma(\omega)$ denotes the second quantization of
multiplication with $\omega(k)=|k|$ in $L^2(\R^3\times\{1,2\})$. See
\cite{reesim:fou}, X.7, for the notation ${\mathrm d}\Gamma(\cdot)$ and for an
introduction to second quantization. The creation and annihilation 
operators $a^*(h)$ and
$a(h)$, for $h\in L^2(\R^3\times\{1,2\})$, are densely defined, closed
operators with $a^*(h)=a(h)^{*}$ and with 
$$
     [a^{*}(h)\Psi]^{(n)} = \sqrt{n}S_n (h\otimes\Psi^{(n-1)})
$$
for vectors $\Psi=(\Psi^{(0)},\Psi^{(1)},\ldots)$ from the subspace $D(N_{f}^{1/2})$.
Here, $\Psi^{(n)}$ denotes the $n$-photon component
of $\Psi$.

The system studied in this paper is composed of a non-relativistic,
(spinless) quantum mechanical, charged particle (the electron), and the
quantized radiation field which is coupled to the electron by minimal 
substitution. In addition, there is an external potential $V$, which may be due to a static nucleus.
The Hilbert space is thus the tensor product
$$
    \cH:= L^2(\R^3)\otimes \F,
$$
and the Hamiltonian is of the form
\begin{align}\label{ham}
   H_{\alpha}&=(\fp+\alpha^{\frac{3}{2}} \fA(\alpha \fx))^2 + V +
   H_f\\
    &= H_0 + W, \nonumber
\end{align}
where $H_0=\hel+H_f$, $\hel=-\Delta+V$, and $W=H_{\alpha}-H_0$. The
quantized vector potential $\fA(\alpha \fx)$, for each $\fx\in\R^3$, is a triple of
self-adjoint operators, each of which is a sum of a creation and an
annihilation operator. Explicitly,
\begin{equation}
    \fA(\alpha\fx) = a(\fG_{\fx})+ a^{*}(\fG_{\fx}), \qquad
    \fG_{\fx}(\fk,\lambda)
    := \frac{\kappa(\fk)}{\sqrt{2|\fk|}}
   \feps(\fk,\lambda)e^{-i\alpha\fk\cdot\fx}, \label{gl1-1}
\end{equation}
where $\feps(\fk,\lambda)\in\R^3$, $\lambda=1,2$, are orthonormal
polarization vectors perpendicular to $\fk$, and
$\kappa$ is an ultraviolet cutoff chosen from the
space $\cS(\R^3)$ of rapidly decreasing functions. No infrared
cutoff is needed. Here and henceforth, the position of the
electron, $\fx\in\R^3$, and the wave vector of a photon, $\fk\in\R^3$, are
dimensionless and related to the corresponding dimensionfull
quantities $\fX,\fK$ by $\fX=(a_0/2)\fx$ and $\fK= (2\alpha/a_0)\fk$, where
$a_0:=\hbar^2/me^2$ is the Bohr-radius, $m>0$ is the mass of the
particle, $e$ its charge, and $\alpha=e^2/\hbar c$ is the fine
structure constant. It follows that $\fX \cdot \fK=\alpha \fx \cdot \fk$, 
and in units where $\hbar$, $c$, and four times the Rydberg energy $2m\alpha^2$
are equal to unity, the Hamiltonian of a one-electron atom with static
nucleus at the origin takes the form
\eqref{ham} with $V(\fx)=-Z/|\fx|$, $Z$ being the atomic number of the
nucleus. For simplicity, we confine ourselves, in this introduction, to
this particular potential $V$. In nature, $\alpha\approx 1/137$, but
in this paper $\alpha$ is treated as a free parameter that can assume
any non-negative value.

For all $\alpha\geq 0$, the 
Hamiltonian $H_{\alpha}$ is self-adjoint on
$D(H_0)$ and its spectrum $\sigma(H_{\alpha})$ is a half-axis
$[E_{\alpha},\infty)$ \cite{Hiroshima2002,HH}. Moreover,
$$
     E_{\alpha}:=\inf\sigma(H_{\alpha})
$$
is an eigenvalue of $H_{\alpha}$, and, at least for $\alpha$
sufficiently small, this eigenvalue is simple \cite{BFS2,GLL}. We use
$\Phi_{\alpha}$ to denote a normalized eigenvector associated with
$E_{\alpha}$. Another important point in the spectrum of $H_{\alpha}$
is the ionization threshold $\Sigma_{\alpha}$, which, for our system, is given by 
$\Sigma_{\alpha} = \inf\sigma(H_{\alpha}-V)$. In a state vector from
the spectral subspace $\Ran\one_{(-\infty,\Sigma_{\alpha})}(H_{\alpha})$, the electron is
exponentially localized in the sense that
\begin{equation}\label{exp}
      \|e^{\beta|x|}\one_{(-\infty,\Sigma_{\alpha}-\eps]}(H_{\alpha})\| <\infty
\end{equation}
for all $\beta$ with $\beta^2<\eps$ \cite{Gri2004}.

The phenomenon of photo-ionization can be considered as a
scattering process, where in the limit $t\to-\infty$, the atom in
its ground state is targeted by a (finite) number of
asymptotically free photons, while in the limit $t\to\infty$ the
atom is ionized in a sense to be made precise. We begin by discussing \emph{incoming
scattering states} and their properties. To this end it is convenient to introduce the space
$L^2_{\omega}(\R\times\{1,2\})$ of all those $f\in
L^2(\R\times\{1,2\})$ for which 
\begin{equation}\label{omega-norm}
   \|f\|^2_{\omega}:=\sum_{\lambda=1,2} \, \int
     |f(\fk,\lambda)|^2 \big(1+\omega(\fk)^{-1}\big) d^3k <
     \infty.
\end{equation}
Given $f\in L^2_{\omega}(\R\times\{1,2\})$, the asymptotic creation
operator $a^{*}_{-}(f)$ is defined by 
\begin{equation}\label{intro1}
   a_{-}^{*}(f)\Psi := \lim_{t\to -\infty}
   e^{iH_{\alpha}t}a^{*}(f_{t})e^{-iH_{\alpha}t}\Psi,\qquad f_{t}:=e^{-i\omega
   t}f,
\end{equation}
and its domain is the space of all vectors $\Psi\in D(|H_{\alpha}|^{1/2})$ for which the limit \eqref{intro1}
exists. This is know to be the case, e.g., for the ground state
$\Psi=\gs$. Moreover, it is known that $a_{-}^{*}(f_1)\cdots
a_{-}^{*}(f_n)\gs$ is well defined and that 
\begin{eqnarray}
    \lefteqn{e^{-iH_{\alpha}t}a_{-}^{*}(f_1)\cdots
    a_{-}^{*}(f_n)\gs}\nonumber\\
    &=& a^{*}(f_{1,t})\cdots a^{*}(f_{n,t})e^{-iH_{\alpha}t}\gs
    + o(1),\qquad (t\to -\infty),\label{intro3}
\end{eqnarray}
whenever $f_i,\omega f_i\in L^2_{\omega}(\R\times\{1,2\})$ for all
$i=1,\ldots,n$ \cite{GZ}. 
By \eqref{intro3}, 
$a_{-}^{*}(f_1)\cdots a_{-}^{*}(f_n)\gs$ describes a scattering state, which, in the limit
$t\to -\infty$ is composed of the atom in its ground state and $n$
asymptotically free photons with wave functions $f_1,\ldots,f_n$. 
Results analogous to those on $a_{-}^{*}(f)$ hold true for
the asymptotic annihilation operators $a_{-}(f)$ \cite{GZ}.

The asymptotic annihilation and creation operators satisfy the usual
canonical commutation relations: e.g.
\begin{equation}\label{CCR}
   [a_{-}(f),a_{-}^{*}(g)] = \sprod{f}{g}
\end{equation}
for all $f,g\in L^2_{\omega}(\R^3\times\{1,2\})$. Moreover, the ground state $\gs$ is a vacuum vector for
asymptotic annihilation operators in the sense that
\begin{equation}\label{vacuum}
   a_{-}(f)\gs = 0\qquad\text{for all}\ f\in L^2_{\omega}(\R^3).
\end{equation}
Hence, if $\uf=(f_1,m_1,\ldots,f_n,m_n)\in
\big[L^2(\R^3\times\{1,2\})\times\N\big]^n$ with
$\sprod{f_i}{f_j}=\delta_{ij}$, then it follows from \eqref{CCR} and \eqref{vacuum}
that
\begin{equation}\label{in-state}
   a_{-}^{*}(\uf)\gs:=\prod_{k=1}^n \frac{1}{\sqrt{m_k!}} a_{-}^{*}(f_k)^{m_k}\gs
\end{equation}
is a normalized vector in $\cH$. All these properties of 
$a_{-}(f),\, a_{-}^{*}(f)$ hold mutatis mutandis for
the asymptotic operators $a_{+}(g),\, a_{+}^{*}(g)$ defined in terms of the limit $t\to+\infty$.

We are interested in the probability that $e^{-iH_{\alpha}t} a_{-}^{*}(\uf)\gs$
describes an ionized atom in the distant future, but we are not 
interested in the asymptotic state of the
electron or the radiation field in the limit $t\to+\infty$. We therefore shall not attempt to
construct outgoing scattering states describing an ionized atom, which
is a difficult open problem. Instead we base our computation of the
probability of ionization on the following reasonable assumption: the atom
described by  $e^{-iH_\alpha t} a_{-}^{*}(\uf)\gs$ is either ionized in the
limit $t\to\infty$, or else, in that limit, it \emph{relaxes to the ground state} in the
sense that  $e^{-iH_\alpha t} a_{-}^{*}(\uf)\gs$, for $t$ large enough, is well approximated
by a linear combination of vectors of the form  
\begin{equation}
    a^{*}(g_{1,t})\ldots a^{*}(g_{n,t})e^{-iE_{\alpha} t}\gs.
\end{equation}
More precisely, relaxation to the ground state occurs if
$a_{-}^{*}(\uf)\gs$ belongs to the closure of the span of all vectors
of the form 
$$
    a^{*}_{+}(g_{1})\ldots a^{*}_{+}(g_{n})\gs = \lim_{t\to+\infty} e^{iH_{\alpha}t}a^{*}(g_{1,t})\ldots a^{*}(g_{n,t})e^{-iE_{\alpha}t}\gs,
$$
with $g_i,\omega g_i\in L^2_\omega(\R^3\times\{1,2\})$. Let $\hap$ denote this space and let $\pap$
be the orthogonal projection onto $\hap$. Then $\|\pap
a_{-}^{*}(\uf)\gs\|^2$ is the probability for relaxation to the ground
state and 
\begin{equation}\label{ion-prob}
    1- \|\pap a_{-}^{*}(\uf)\gs\|^2 = \|(1-\pap) a_{-}^{*}(\uf)\gs\|^2
\end{equation}
is the probability of ionization.

The assumption that relaxation to the ground state is the only
alternative to ionization, is motivated by the conjecture of \emph{asymptotic completeness for
Rayleigh scattering}, which is the property, that every vector $\Psi\in\cH$ describing a bound state in the sense that 
$\sup_{t}\|e^{\eps|x|}e^{-iH_{\alpha}t}\Psi\| <\infty$ for some
$\eps>0$, will relax the ground state in the limit
$t\to\infty$.  In view of \eqref{exp}, asymptotic completeness for Rayleigh scattering implies that $\hap \supseteq \one_{(-\infty,\Sigma_{\alpha})}(H_{\alpha})$, which can be proven for simplified models of atoms \cite{Spohn,DG1,Gerard2002,FGS2}.

The following two theorems will allow us to compute \eqref{ion-prob}.

\begin{theorem}\label{thm1}
Suppose that $f_1,\ldots,f_n\in L^2(\R^3\times \{1,2\})$ where $\sum_{\lambda=1}^2 \feps(\cdot, \lambda) f_i(\cdot,\lambda)$ belongs to 
$C_0^2(\R^3\backslash\{\fn\}, \C^3)$ for each $i$, and let $\uf=(f_1,\ldots,f_n)$. Then:
\begin{equation}\label{reduction}
    a_{-}^{*}(\uf)\gs =  a_{+}^{*}(\uf)\gs
    -i\alpha^{3/2}\int_{-\infty}^{\infty} 2\fp(s)\ph_{el}\otimes
    [\fA(\fn,s),a^{*}(\uf)]\Omega\,ds + \cO(\alpha^{5/2})
\end{equation}
where $\fp(s)=e^{iH_{el}s}\fp e^{-iH_{el}s}$ and $\fA(\fn,s)=e^{iH_fs}\fA(\fn)e^{-iH_fs}$.
\end{theorem}
The first term of \eqref{reduction} gives no contribution to
the ionization probability \eqref{ion-prob} because
$a_{+}^{*}(\uf)\gs\in \cH_{+}^{\alpha}$. The second term is
proportional to $\alpha^{3/2}$ and it is due to scattering processes
where one of the $n$ photons $f_1,\ldots,f_n$ is absorbed.
The remainder terms are of order $\cO(\alpha^{5/2})$
and stem from the dipole approximation $\fA(\alpha\fx)\to \fA(\fn)$, from
dropping $\alpha^3\fA(\fx)^2$ and from
ignoring processes of higher order in $\alpha^{3/2}$.
To isolate the contribution of order $\alpha^3$ from \eqref{ion-prob}
using \eqref{reduction}, we need:

\begin{theorem}\label{thm2}
Suppose that $\hap \supseteq \one_{(-\infty,\Sigma_{\alpha})}(H_{\alpha})$ for
$\alpha$ in a neighborhood of $0$, and suppose that $\hel$ has only negative
eigenvalues. Then
\begin{equation}\label{pap-pp}
\lim_{\alpha \to 0} \pap = \one_{pp}(\hel)
\end{equation}
in the strong operator topology.
\end{theorem}

Combining Theorem~\ref{thm1} and Theorem~\ref{thm2} we see that 
\begin{eqnarray*}
   \lefteqn{\|(1-\pap) a_{-}^{*}(\uf)\gs\|^2 =
     \|(1-\pap)\Big(a_{-}^{*}(\uf)-a_{+}^{*}(\uf)\Big)\gs\|^2}\\
 &=& \|\one_{c}(\hel) \Big(a_{-}^{*}(\uf)-a_{+}^{*}(\uf)\Big)\gs\|^2 +o(\alpha^3)\\
 &=&  \alpha^3 \|\one_{c}(\hel)\int_{-\infty}^{\infty}\fp(s)\ph_{el}\otimes
 [\fA(\fn,s),a^{*}(\uf)]\Omega\,ds\|^2 + \co(\alpha^3)
\end{eqnarray*}
where $\one_{c}(\hel)=1-\one_{pp}(\hel)$, and where the second
equation is justified by the $\alpha$ dependence of 
$a_{-}^{*}(\uf)\gs-a_{+}^{*}(\uf)\gs$ as given by \eqref{reduction}.
We are now going to express the coefficient of $\alpha^3$ in terms of
generalized eigenfunctions of $\hel$, which makes it explicitly
computable in simple cases. A general and sufficient condition for the existence of a complete set of
generalized eigenfunctions is the existence and completeness of  
a (modified) wave operator $\Omega_{+}$ associated
with $\hel$. This condition is satisfied for our choice of $V$. It
means that there exists an isometric operator
$\Omega_{+}\in\cL(\Hel)$ with $\Ran\Omega_{+}=\one_{c}(\hel)\Hel$
and $\hel\Omega_{+}=\Omega_{+}(-\Delta)$.
In particular, the singular continuous spectrum of $\hel$ is
empty. Given the wave operator $\Omega_{+}$ and the fact that
$(\hel-i)^{-1}\expect{\fx}^{-2}$ is a Hilbert-Schmidt operator, it
is  easy to establish existence of generalized eigenfunctions
$\ph_{\fq}$, $\fq\in\R^3$, of $\hel$ with the following properties
\cite{PS}:
\begin{itemize}
\item[(i)] The function $(\fx,\fq)\mapsto \expect{\fq}^{-2}
\expect{\fx}^{-2}\ph_{\fq}(x)$ is square
integrable on $\R^3\times\R^3$, in particular
$\expect{\fx}^{-2}\ph_{\fq}\in L^2(\R^3)$ for almost every
$\fq\in\R^3$. $\expect{\fx}:=(1+|\fx|^2)^{1/2}$.  
\item[(ii)] If $\psi\in D(|\fx|^2)$ then
\begin{equation}\label{gef1}
    \|\one_{ac}(\hel)\psi\|^2 = \int_{\R^3}|\sprod{\ph_{\fq}}{\psi}|^2d^3q
\end{equation}
\item[(iii)] If $F:\R\to\C$ is a Borel function, $\psi\in D(|\fx|^2)\cap
D(F(\hel))$, and $F(\hel)\psi\in D(|\fx|^2)$, then
\begin{equation}\label{gef2}
  \sprod{\ph_{\fq}}{F(\hel)\psi} = F(\fq^2)\sprod{\ph_{\fq}}{\psi}
\end{equation}
for almost every $\fq\in\R^3$.
\end{itemize}
In (ii) and (iii) we use $\sprod{\ph_{\fq}}{\psi}$ to denote the
integral $\int \overline{\ph_{\fq}(\fx)}\psi(\fx)\,d^3x$, 
which is well defined by (i) and by the assumption
$\psi\in D(|\fx|^2)$.

The Theorem~\ref{thm1} in conjunction with (i)-(iii) implies the
following theorem, which is our main result specialized to the case of
only one asymptotic photon in the incident scattering state.

\begin{theorem}\label{thm3}
For all $f\in L^2(\R^3 \times \{1,2\})$ with 
$\sum_{\lambda=1}^2 \feps(\cdot, \lambda) f(\cdot,\lambda) \in
C_0^2(\R^3\backslash\{\fn\}, \C^3)$,
\begin{eqnarray}\label{intro5}
     \lefteqn{\left\|\one_{ac}(\hel)\big(a_{-}^{*}(f)\Phi_{\alpha}-a_{+}^{*}(f)
\Phi_{\alpha}\big)\right\|^2}\nonumber\\
 &=& \alpha^3 \int_{\R^3} d^3q \left|\sprod{\ph_{\fq}}{\fx\ph_{el}}\cdot
 \int_{-\infty}^{\infty}e^{i(\fq^2-E_0)t}\sprod{\Omega}{\fE(\fn,t)a^{*}(f)\Omega}dt\right|^2 +
\cO(\alpha^4)
\end{eqnarray}
as $\alpha\to 0$. Here, $\fE(\fn,t)=-i[H_f,\fA(\fn,t)]$, $\ph_{el}$ is a normalized ground state
of $\hel$ and $\ph_{\fq}$, $\fq\in\R^3$, is any family of
generalized eigenfunction of $\hel$ with properties (i)-(iii) above.
\end{theorem}

The expression \eqref{intro5} for the ionization probability can
be understood, \emph{on a formal level}, by first order, time-dependent perturbation
theory. To this end one considers the transitions $\ph_{el}\otimes f\mapsto
\ph_{\fq}\otimes \Omega$, for fixed $\fq\in\R^3$, in the \emph{interaction
picture} defined by $H_0$. Then the time-evolution of state vectors is 
generated by the time-dependent interaction operator $W(t)=e^{iH_0t}We^{-iH_0t}
= 2\alpha^{3/2}\fp(t)\cdot \fA(\alpha \fx,t) + \alpha^3 \fA(\alpha
\fx,t)^2$ with $\fp(t)=e^{i\hel t} \fp e^{-i\hel t}$ and $\fA(\alpha\fx,t) =
e^{iH_0t}\fA(\alpha\fx)e^{-iH_0 t}$. In the computation of the
transition amplitude to the order $\alpha^{3/2}$ one drops $\alpha^3 \fA(\alpha\fx,t)^2$ and one replaces 
$\fA(\alpha \fx,t)$ by $\fA(\fn,t)$, which is
known as the dipole approximation. Then, an integration by parts using
that 
$$
   2\fp(t) = \frac{d}{dt}\fx(t),\qquad -\frac{\partial}{\partial t}\fA(\fn,t) = \fE(\fn,t),
$$  
leads to a result for the transition amplitude which agrees with the
expression in \eqref{intro5}
whose modulus squared is integrated over $\fq\in\R^3$. The
Theorem~\ref{thm3} and its proof justify this formal
derivation and the use of the dipole approximation. Note that
$\alpha \fx =\fX$, hence the ionization probability is of order
$\alpha^3$ rather than of order $\alpha$, as a formal computation, similar to the
one above, in dimension-full quantities would suggest.

We prove a more general result than Theorem \ref{thm3}, where the
incoming scattering state may contain several asymptotic photons,
and where the external potential $V$ is taken from a large class
of long range potentials. In the case where the asymptotic state
at $t=-\infty$ is of the form \eqref{in-state}
and each of the photons $f_1,\ldots,f_n\in L^2(\R^3\times\{1,2\})$ satisfies
the hypotheses of Theorem~\ref{thm3}, in addition to
$\sprod{f_i}{f_j}=\delta_{ij}$, our result says that
\begin{equation}\label{intro5b}
   \left\|\one_{ac}(\hel)\big(a_{-}^{*}(\uf)\Phi_{\alpha}-a_{+}^{*}(\uf)
\Phi_{\alpha}\big)\right\|^2 = \alpha^3\sum_{l=1}^n m_l P^{(3)}(f_l) + \cO(\alpha^4)
\end{equation}
with 
\begin{equation}
P^{(3)}(f_l):=\int\limits_{\R^3} d^3q \left|\sprod{\ph_{\fq}}{\fx\ph_{el}}\cdot
 \int_{-\infty}^{\infty}\sprod{\Omega}{\fE(\fn,t)a^{*}(f_l)\Omega}e^{i(\fq^2-E_0)t}dt\right|^2.
\label{intro6}
\end{equation}
The integral with respect to $t$ in \eqref{intro6} can be
computed explicitly in terms of $f_l$ and $\fG_0$, and it gives
\begin{eqnarray*}
\lefteqn{\int_{-\infty}^{\infty}e^{i(\fq^2-E_0)t}\sprod{\Omega}{\fE(\fn,t)a^{*}(f_l)\Omega}dt}\\
&=&i\pi\int_{|\fk|=\fq^2-E_{0}}
    \kappa(\fk)\sqrt{2|\fk|}\sum_{\lambda=1,2}\feps_{\lambda}(\fk,\lambda)f_l(\fk,\lambda)
d\sigma(\fk),
\end{eqnarray*}
where $d\sigma(\fk)$ is the surface measure of the sphere
$\{\fk \in \R^3: |\fk|=\fq^2-E_{0}\}$ in $\R^3$. The integration over
the spheres with $|\fk|=\fq^2-E_0$ expresses the conservation of
energy in the scattering process, and the additivity \eqref{intro5b} of the ionization
probability with respect to the incoming photons corresponds to the experimental
fact, that the number of photo-electrons is proportional to the
intensity of the incoming radiation.

In Section~\ref{s4-3} we give a second derivation of $\alpha^3
P_3(\uf)$ based on a space-time analysis
of the ionization process. This approach, in a slightly different form, was
introduced in the papers \cite{BKZ,Z}, and does not assume
asymptotic completeness of Rayleigh scattering.

The existence of outgoing scattering state describing an ionized atom
and an electron escaping to spacial infinity is a
difficult open problem in the model described above. Only for $V=0$
such states have been constructed so far \cite{Pizzo2003,ChenFroePizzo2007}. Hence it is not
possible yet to study the ionization probability based on transition
probabilities between asymptotic states. 

Previously ionization by {\it quantized} fields was investigated in 
\cite{BKZ,FM1,FM2,Z}. \cite{BKZ} and \cite{Z} are precursors of the
present paper on simpler models of atoms and the ionization
probability defined in a different, but equivalent way. In
\cite{FM1,FM2} it is shown that a thermal quantized field leads to
ionization in the sense of absence of an equilibrium state of atom and
field. \\
There is a large host of mathematical results on ionization by 
{\it classical} electric fields: Schrader and various coauthors study
the phenomenon of stabilization by providing upper and lower bounds on
the ionization probability, see \cite{EKS,FKS,Laser} and the
references therein. They use the Stark-Hamiltonian with a time
dependent electric field $\mathcal{E}(t)$ that vanishes unless $0 \leq
t \leq \tau < \infty$. Lebowitz and various coauthors compute the
probability of ionization by an electric field that is periodic in
time; see \cite{MR2436498,MR1768630} and references therein. Most
of these papers study one-dimensional Schr\"odinger operators with a single
bound state that is produced by a $\delta$-potential.
Ionization in a three-dimensional model with a $\delta$-potential is
studied in \cite{CDFM}.

\noindent
\emph{Acknowledgment:} M.G. thanks Vadim Kostrykin for pointing out
that it is advantageous to define ionization as the opposite of binding.   

%%%%%%%%%%%%%%%%%%%%%%%%%%%%%%%%%%%%%%%%%%%%%%%%
%%%%%%%%%%%%%%%%%%%%%%%%%%%%%%%%%%%%%%%%%%%%%%%%
%%%%%%%%%%%%%%%%%%%%%%%%%%%%%%%%%%%%%%%%%%%%%%%%

%%%%%%%%%%%%%%%%%%%%%%%%%%%%%%%%%%%%%%%%%%%%%%%%
%%%%%%%%%%%%%%%%%%%%%%%%%%%%%%%%%%%%%%%%%%%%%%%%
%%%%%%%%%%%%%%%%%%%%%%%%%%%%%%%%%%%%%%%%%%%%%%%%
\section{Notations and Hypotheses}
\label{ch2}
\setcounter{equation}{0}
%%%%%%%%%%%%%%%%%%%%%%%%%%%%%%%%%%%%%%%%%%%%%%%%
%%%%%%%%%%%%%%%%%%%%%%%%%%%%%%%%%%%%%%%%%%%%%%%%
%%%%%%%%%%%%%%%%%%%%%%%%%%%%%%%%%%%%%%%%%%%%%%%%

{For} easy reference, we collect in this section the definitions, our
notations and all hypotheses. As usual, $L^2(\R^3 \times \{1,2\})$ denotes
the space of square integrable functions $f:\R^3 \times \{1,2\} \to \C$
with inner product
\[ \lkl f,g \rkl := \sum_{\lambda=1,2} \, \int\limits_{\R^3}
     \overline{f(\fk,\lambda)}g(\fk,\lambda) d^3k. \]
We recall from the introduction that $L^2_{\omega}(\R^3 \times \{1,2\})$ consists of those
functions $f \in L^2(\R^3 \times \{1,2\})$ for which the norm
$\|f\|_{\omega}$ defined in \eqref{omega-norm} is finite.
Regularity assumptions will be imposed on the vector-valued function
\begin{equation}\label{T-photon}
(\feps f)(\fk):=\sum_{\lambda=1}^2 \feps(\fk,\lambda) f(\fk,\lambda),
\end{equation}
rather than on on $f(\cdot,1)$ and $f(\cdot,2)$. It is useless to impose smoothness conditions on $f(\cdot,\lambda)$ because it is \eqref{T-photon} that matters and because the 
polarization vectors $\feps(\fk,1)$ and $\feps(\fk,2)$ are necessarily discontinuous. On the other hand, every square integrable function $f:\R^3\to\C^3$ with 
$\fk\cdot f(\fk)$, for a.e. $\fk\in\R^3$, can be approximated, in the $L^2$-sense, by smooth functions of the form \eqref{T-photon}.

%The Hilbert space $L^2(\R^3 \times \{1,2\})$ is isomorphic to the space 
%$L^2_{T}(\R^3)$ of square integrable functions $f: \R^3 \to \C^3$ with
%$\fk \cdot f(\fk)=0$ for almost every $\fk \in \R^3$. In fact, each choice of (normalized) polarization
%vectors $\feps(\fk,1), \feps(\fk,2)\perp\fk$ defines an isomorphism 
%$\feps : L^2(\R^3 \times \{1,2\}) \to L^2_t(\R^3)$ by the equation
%\[(\feps f)(\fk):=\sum_{\lambda=1}^2 \feps(\fk,\lambda) f(\fk,\lambda). \]
%Note that $\feps f$ can be approximated in the norm of $L^2_t(\R^3)$
%by smooth functions even though the polarization vectors
%$\feps(\fk,1)$ and $\feps(\fk,2)$ are necessarily discontinuous. This
%is the reason for introducing the space $L^2_t(\R^3)$. 

It is convenient to collect a family $f_1,...,f_N \in L^2(\R^3 \times
\{1,2\})$ of photon wave functions in an $N$-tupel
$\uf=(f_1,...,f_N)$. We define
\begin{eqnarray*}
a(\uf)&:=& a(f_1) \cdots a(f_N) \\
a^*(\uf)&:=&a^*(f_1) \cdots a^*(f_N). 
\end{eqnarray*}
This should not lead to confusion with \eqref{in-state}, where $\uf$ also includes occupation numbers.
{For} the various parts of the interaction operator $W=H_{\alpha}-H_0$,
we use the notations
\begin{eqnarray*}
\Wdip&:=&2 \fp \cdot \fA(\fn), \\
W^{(1)}&:=&2 \fp \cdot \fA(\alpha \fx), \\
W^{(2)}&:=& \fA(\alpha \fx)^2. 
\end{eqnarray*}
It follows that
$$
  W=\alpha^{\frac{3}{2}} W^{(1)}+\alpha^3 W^{(2)}= 
  \alpha^{\frac{3}{2}}\Wdip + \cO(\alpha^{\frac{5}{2}})
$$
where the last equation is purely formal, but we shall give it a
rigorous meaning in this paper. The Hamiltonian
\[H_{\alpha}=H_0+W\]
is self-adjoint on the domain of $-\lap+H_f$ provided that $V$ is
infinitesimally operator bounded with respect to $-\lap$, 
\cite{HH,Hiroshima2002}.
This is the case, e.g., if $V$ is the sum of Coulomb potentials due to
static nuclei; all our results are valid for such $V$. Nonetheless,
it is useful to identify the properties of $V$ that are essential for
our analysis. From now on, we shall only assume the following
hypotheses on $V$:

\medskip
\noindent\textbf{Hypotheses:}
\emph{Both $V$ and $\fx\cdot\nabla V$ belong to $\in L^2_{loc}(\R^3)$, $\lim_{|\fx|\to\infty}V(\fx)=0$, and there
exist constants $\mu>0$ and $R>0$ such that for $|\beta|=1,2$ we
have
$$
     |\partial_{\fx}^{\beta}V(\fx)| \leq |\fx|^{-|\beta|-\mu},\qquad\text{if}\ 
|\fx|>R.
$$
Moreover, $E_0:=\inf\sigma(\hel)<0$. We define $e_1:=\inf(\sigma(\hel) \backslash \{E_0\})$.}

\medskip
{From} these Hypotheses it follows that
$\sigma_{ess}(\hel)=[0,\infty)$, that
$\sigma_{sc}(\hel)=\emptyset$ and that $E_0$ is a simple
eigenvalue. In fact, the decay assumptions on $V$
imply long-range asymptotic completeness \cite{DerezinskyGerard1997a}, which is what
we use to infer the existence of a complete set of generalized
eigenfunctions. All this remains true if a singular short-range potential is added to $\hel$.

The time evolution of an operator $B$ in the \emph{interaction picture} will be denoted by $B(t)$, that is, 
$$
        B(t) := e^{iH_0 t}Be^{-iH_0 t},
$$
and $B_t:=B(-t)$. Note that $\fp(t)=e^{i\hel t}\fp e^{-i\hel t}$, $\fA(\fn,t)=e^{iH_f t}\fA(\fn) e^{-iH_f t}$ and that $a^{\#}(\uf_t)=e^{-iH_0 t}a^{\#}(\uf)e^{iH_0 t}=a^{\#}(\uf)_t$.

%%%%%%%%%%%%%%%%%%%%%%%%%%%%%%%%%%%%%%%%%%%%%%%%
%%%%%%%%%%%%%%%%%%%%%%%%%%%%%%%%%%%%%%%%%%%%%%%%
%%%%%%%%%%%%%%%%%%%%%%%%%%%%%%%%%%%%%%%%%%%%%%%%
\section{Commutator estimates and scattering states} 
\label{sec2}
\setcounter{equation}{0}

The main purpose of this section is to establish bounds on the
commutators $[W^{(j)},a^*(\uf_t)]$ applied to $\Phi_\alpha$ for $W^{(j)}\in\{W^{(1)}, W^{(2)},\Wdip\}$.
We are interested in the decay as
$|t|\to\infty$ and in the dependence on $\alpha$. Typically, our
estimates are valid for $\alpha\leq \tilde\alpha$, where
$\tilde\alpha$ is defined in Proposition~\ref{pA-4}.
As a simple application of our decay estimates in $t$, we will obtain existence
of the scattering states
$$
 a^*_{\pm} (\uf) \gs = \lim_{t \to \pm} e^{itH_{\alpha}} a^*(\uf_t) 
e^{-itH_{\alpha}}\gs,
$$
which was already established in \cite{GZ} in larger generality.
Here $\uf_t=(f_{1,t},...,f_{N,t})$ and $f_{j,t}:=e^{-it\omega} f_j$.
Given $l \in \{1,...,N\}$, we write
\begin{eqnarray*}
  a^*(\uf_{[l],t}) &:=& a^*(f_{1,t}) \cdots a^*(f_{l-1,t}) \fA(\alpha \fx)
a^*(f_{l+1,t}) \cdots a^*(f_{N,t}),\\
   a^*(\uf_{(l),t})&:=& a^*(f_{1,t}) \cdots a^*(f_{l-1,t})
a^*(f_{l+1,t}) \cdots a^*(f_{N,t}).
\end{eqnarray*}
For $\fx\in\R^3$, $\expect{\fx}:=(1+|\fx|^2)^{1/2}$.

%-------------------------------------------------------------------------------------

\begin{lemma}\label{lm:stat-phase}
Suppose that $f\in L^2(\R^3\times\{1,2\})$ with 
$\feps f\in C_0^n(\R^3 \backslash \{\fn\},\C^3)$ for a given $n\in\N$. Then there exists 
is a constant $c_{1,n}=c_{1,n}(f)$ such that
\begin{eqnarray}
   |\lkl \fG_{\fn},f_t \rkl| &\leq & c_{1,n} \frac{1}{1+|t|^n},
\label{stat-phase1}\\
   |\lkl \fG_{\fx},f_t \rkl| &\leq &
   c_{1,n} \frac{1+(\alpha|x|)^n}{1+|t|^n}\qquad\text{for all}\ \fx\in\R^3,
\label{stat-phase2} \\
   |\lkl \fG_{\fx}- \fG_{\fn},f_t \rkl| &\leq & c_{1,n} \frac{\alpha |\fx|
\expect{\alpha \fx}^n}{1+|t|^n}\qquad\text{for all}\ \fx\in\R^3.
\label{stat-phase3}
\end{eqnarray}
\end{lemma}

\begin{proof}
Estimate \eqref{stat-phase1} follows from \eqref{stat-phase2}. We next
prove  \eqref{stat-phase2}. By a stationary phase analysis of 
\begin{equation} 
  \sprod{\fG_{\fx}}{f_{t}} =
  \int\limits_{\R^3} d^3k \frac{\kappa(\fk)}{\sqrt{2\omega(\fk)}}
  e^{i\alpha \fk \cdot \fx - it \omega(\fk)} (\feps f)(\fk) \label{gl3-5}
\end{equation}
we obtain  
$|\lkl \fG_{\fx},f_t \rkl | \leq C_n |t|^{-n}$ for
$\alpha |\fx| \leq |t|/2$, \cite{ReedSimon3}  Theorem XI.14. It
follows that
\begin{eqnarray*} 
|\lkl \fG_{\fx}, f_{t} \rkl | \one_{\{2\alpha |\fx| \leq |t|\}}
&\leq& \frac{C_n}{|t|^n}\\
|\lkl \fG_{\fx}, f_{t} \rkl | \one_{\{2\alpha |\fx| > |t|\}}
&\leq& C \left(\frac{2\alpha|\fx|}{|t|}\right)^n \nn
\end{eqnarray*}
where  $C:=\sup_{t \in \R,\,\fx \in \R^3} |\lkl \fG_{\fx},f_t \rkl| <
\infty$. This proves \eqref{stat-phase2}.
To prove \eqref{stat-phase3} we write
\[\lkl \fG_{\fx}-\fG_\fn,f_t \rkl = \int\limits_{\R^3}
e^{-it\omega(\fk)} \fF_{\fx}(\fk) d^3k \]
where
\[\fF_{\fx}(\fk)=i \alpha\fk \cdot \fx
\frac{\kappa(\fk)}{\sqrt{2\omega(\fk)}}
(\feps f)(\fk) g(\alpha \fk \cdot \fx) \]
and $g: \R \to \C$ denotes the real-analytic function given by
$g(s)=(e^{is}-1)/(is)$ for $s \not =0$. $g$ and all its derivatives
are bounded, and by assumption on $f$, $\fF_{\fx} \in
C_0^{\infty}(\R^3 \backslash \{\fn\}, \C^3)$ for each $\fx$. It follows
that
\[\sup_{\fx,\fk \in \R^3,\, \fx\not=\fn} 
\Big| \partial_{\fk}^{\beta} \fF_{\fx}(\fk) \Big| 
|\fx|^{-1} \lkl \alpha \fx \rkl^{-|\beta|} < \infty, \]
which implies \eqref{stat-phase3}, again by stationary phase arguments.
\end{proof}

%%%%%%%%%%%%%%%%%%%%  erstes Kommutatorlemma   %%%%%%%%%%%%%%%%%%%%%%%%%%%%%%%%

\begin{lemma} \label{l3-1}
Suppose that $\feps f_1,..., \feps f_N \in
C_0^{n}(\R^3 \backslash \{\fn\},\C^3)$ for a given $n\in\N$, and
let $\tilde\alpha$ be defined by Proposition~\ref{pA-4}.
Then there exist constants $\tilde{\alpha}>0$ and
$c_{2,n} = c_{2,n}(\uf)$, such that for all
$\alpha\leq\tilde{\alpha}$, $t\in \R$, and $W^{(j)}\in\{W^{(1)},W^{(2)},\Wdip\}$,
\begin{equation} 
  \big\|\big[W^{(j)},a^*(\uf_t)\big] \gs\big\| \leq \frac{c_{2,n}}{1+|t|^n}. \label{gl3-5a} 
\end{equation}
\end{lemma}
\begin{proof} 
By definition of $a^*(\uf_t)$,
\begin{equation}
[W^{(j)},a^*(\uf_t)] \gs=
\sum_{l=1}^N a^*(f_{1,t}) \cdots a^*(f_{l-1,t}) \big[W^{(j)},a^*(f_{l,t})\big]
a^*(f_{l+1,t}) \cdots a^*(f_{N,t})\gs \label{gl1}
\end{equation}
and by definition of $W^{(1)}$ and $W^{(2)}$
\begin{eqnarray}
\big[W^{(1)},a^*(f_{l,t})\big]&=& 
2 \lkl \fG_{\fx},f_{l,t} \rkl \cdot \fp \label{gl2}\\
\big[W^{(2)},a^*(f_{l,t})\big]&=& 
2 \lkl \fG_{\fx},f_{l,t} \rkl \cdot \fA(\alpha \fx) \label{gl3}
\end{eqnarray}
{From} \eqref{stat-phase2}, \eqref{gl1}, \eqref{gl2}, \eqref{gl3} and 
Lemma \ref{la2} it follows that
\begin{eqnarray} 
\|[W^{(1)},a^*(\uf_t)] \gs \| &\leq& \frac{Nc_n}{1+|t|^n}
\|(H_f+1)^{\frac{N-1}{2}} \lkl \alpha \fx \rkl^n \fp \gs \| \label{gl4}\\
\|[W^{(2)},a^*(\uf_t) ] \gs \| &\leq& \frac{Nc_n}{1+|t|^n}
\|(H_f+1)^{\frac{N}{2}} \lkl \alpha \fx \rkl^n \gs \|
\label{gl3-6} 
\end{eqnarray}
with some constant $c_n$. Thanks to Lemma \ref{la1}, 
these upper bounds are bounded uniformly in $\alpha \leq
\tilde{\alpha}$, $\tilde{\alpha}$ being defined by Proposition \ref{pA-4}.
This proves \eqref{gl3-5a} for $j=1,2$. 
The assertion for $\Wdip$ now follows from $\Wdip=W^{(1)}|_{\fx=\fn}$,
which leads to a bound for
$\|[\Wdip,a^*(\uf_t)]\gs\|$ of the form \eqref{gl4} with $\fx=\fn$.
\end{proof}

%%%%%%%%%%%  einfacher Kommutator von W1 - Wdip / Phi_alpha   %%%%%%%%%%%%%%%%%%%%%%%%%

\begin{proposition} \label{c3-4n}
For all $\feps f_1,...,\feps f_N \in C_0^{2}(\R^3 \backslash \{\fn\}, \C^3)$
there exists a constant $c_3=c_3(\uf)$, such that for all
$\alpha\leq\tilde{\alpha}$ and for all $s \in \R$,
\begin{equation} \Big\| \Big[ W^{(1)}-\Wdip,a^*(\uf_s) \Big] \gs \Big\|
\leq \frac{c_3\alpha}{1+s^2} . 
\end{equation}
\end{proposition}
\begin{proof}
By \eqref{gl1} for $j=1$, \eqref{gl2}, and the corresponding equations
for $\Wdip$
\[\big[W^{(1)}-\Wdip,a^*(\uf_{s})\big] \gs=
2 \sum_{l=1}^N \lkl \fG_{\fx}-\fG_\fn,f_{l,s} \rkl \cdot \fp
a^*(\uf_{(l),s}) \gs\]
where
\begin{eqnarray*}
\|\lkl \fG_{\fx}-\fG_\fn,f_{l,s} \rkl \cdot \fp a^*(\uf_{(l),s}) \gs \|
&\leq& \alpha \frac{c}{1+s^2} \big\| |\fx| \lkl \alpha \fx \rkl^2 
a^*(\uf_{(l),s}) \fp \gs \big\| \\
&\leq&
\alpha \frac{c}{1+s^2} \big\| |\fx|^2 \lkl \alpha \fx \rkl^4 \fp \gs
\big\|^{1/2} \|a(\uf_{(l),s}) a^*(\uf_{(l),s}) \fp \gs \big\|^{1/2}
\end{eqnarray*}
by \eqref{stat-phase3} and the Cauchy-Schwarz inequality. The norms in
the last expression are bounded uniformly in $\alpha \leq
\tilde{\alpha}$ by Lemma \ref{la2} and Lemma \ref{la1}.
\end{proof}

%%%%%%%%%%%  Kommutatoren von  Wdip   %%%%%%%%%%%%%%%%%%%%%%%%%

\begin{lemma} \label{l3-4}
For all $\feps f_1,...,\feps f_N \in C_0^{2}(\R^3 \backslash \{\fn\}, \C^3)$,
there exists a constant $c_4=c_4(\uf) < \infty$, such that for all
$\alpha \leq\tilde{\alpha}$ and $s,t \in \R$
\begin{eqnarray}
\big\|\big[\Wdip_s,a^*(\uf_t)\big](\gs-\go)\big\| &\leq& \frac{c_4 \alpha^{\frac{3}{2}}}
{1+|t-s|^2} \label{gl3-23}\\
\big\|\big[\Wdip_s,a^*(\uf_t)\big]\go\big\| &\leq& \frac{c_4}
{1+|t-s|^2}.  \label{gl3-27b} \\
\big\|\big[W,\big[\Wdip_s,a^*(\uf_t)\big]\big] \gs \| &\leq &
\frac{c_4 \alpha^{\frac{3}{2} }} {1+|t-s|^2} \label{gl3-22}
\end{eqnarray}
\end{lemma}

\begin{proof}
Since $\big[\Wdip_s,a^*(f_{l,t}) \big]=
2 \lkl \fG_\fn,f_{l,t-s} \rkl \cdot \fp_s$, which 
commutes with the creation operators $a^*(f_{i,t})$, 
\begin{eqnarray}
\big[\Wdip_s,a^*(\uf_t)\big]
&=& \sum_{l=1}^N a^*(f_{1,t})\cdots a^*(f_{l-1,t})
\big[\Wdip_s,a^*(f_{l,t}) \big] a^*(f_{l+1,t}) \cdots a^*(f_{N,t})\nn\\
&=& 2 \sum_{l=1}^N a^{*}(\uf_{(l),t}) \lkl \fG_\fn,f_{l,t-s} \rkl\cdot \fp_s\label{gl3-24}
\end{eqnarray}
where $|\lkl \fG_\fn, f_{l,t-s} \rkl| \leq c_l(1+(t-s)^2)^{-1}$ by
\eqref{stat-phase1}. In view of 
Lemma \ref{la2} and Lemma \ref{la6}, this proves \eqref{gl3-23}. 
The proof of \eqref{gl3-27b} is similar.

{From} \eqref{gl3-24} we obtain, that
\[\Big[W,\Big[\Wdip_s,a^*(\uf_t)\Big]\Big] \gs=
2\alpha^{\frac{3}{2}} \sum_{l=1}^N \lkl \fG_\fn,f_{l,t-s} \rkl \cdot
\Big[W^{(1)}+\alpha^{\frac{3}{2}} W^{(2)},a^*(\uf_{(l),t}) \fp_s \Big] \gs. \]
Hence, by \eqref{stat-phase1}, it suffices to show that
$\big\|W^{(j)}a^*(\uf_{(l),t}) \fp_s \gs\big\|$ and
$\big\|a^*(\uf_{(l),t}) \fp_s W^{(j)} \gs\big\|$ 
are bounded uniformly in $t,s$ and $\alpha \leq \tilde{\alpha}$. 
We shall do this for $a^*(\uf_{(l),t}) \fp_s W^{(1)}\gs$ only, the
proofs in the other cases being similar. Let $m\geq
(N-1)/2$. Then
\begin{eqnarray*}
   \lefteqn{\big\|a^*(\uf_{(l),t}) \fp_s W^{(1)}\gs\big\|} \\
  &\leq&
  \sum_{j=1}^3\big\|a^*(\uf_{(l),t})(H_f+1)^{-m}\fp_sp_j(\hel+i)^{-1}(\hel+i)(H_f+1)^mA_j(\alpha\fx)\gs\big\|\\
  &\leq & C\sum_{j=1}^3 \|(\hel+i)(H_f+1)^mA_j(\alpha\fx)\gs\|
\end{eqnarray*}
with a constant $C$, that is finite by Lemma~\ref{la2}. We now want to compare
$\|(\hel+i)(H_f+1)^mA_j(\alpha\fx)\gs\|$ with $\|
A_j(\alpha\fx)(\hel+i)(H_f+1)^m\gs\|$, because the latter norm is
bounded uniformly in $\alpha\leq\tilde\alpha$, by Lemma~\ref{la2} and
by \eqref{glAA-27}. Thus we compute the commutator of
$(\hel+i)(H_f+1)^m$ and $A_j(\alpha\fx)=a^{*}(G_{\fx,j})+a(G_{\fx,j})$ applied to $\gs$. Using
\begin{eqnarray*}
  \big[\hel,a^{*}(G_{\fx,j})\big] &=& \alpha^2 a^{*}(\omega^2
  G_{\fx,j}) -\alpha\sum_{m=1}^3 2 a^{*}(k_m G_{\fx,j})p_m \\
  \big[(H_f+1)^m,a^{*}(G_{\fx,j})\big] &=& {\ds\sum_{l=1}^m
  \binom{m}{l}a^{*}(\omega^{l}G_{\fx,j})(H_f+1)^{m-l}}
\end{eqnarray*}
and similar commutator equations for $a(G_{\fx,j})$, we see that all resulting terms have norms that are bounded, uniformly in
$\alpha\leq\tilde\alpha$, thanks to \eqref{glAA-27} and Lemma~\ref{la2}. 
\end{proof}

For completeness of this paper we now use Lemma \ref{l3-1} to prove
existence of the asymptotic creation and annihilation operators on $\gs$.
More general results can be found in \cite{FGS1,GZ}.

%-------------- existence of scattering states  -------------------------

\begin{proposition} \label{p3-5}
Suppose $\uf=(f_1,...,f_N) \in [L^2_{\omega}(\R^3 \times \{1,2\})]^N$. Then,
for all $\alpha \leq\tilde{\alpha}$,
\begin{equation} 
  a^*_{\pm}(\uf) \gs := \lim_{t \to\pm\infty}
  e^{iH_{\alpha}t} a^*(\uf_t)e^{-iH_{\alpha}t} \gs
\end{equation}
exists, and 
\begin{equation} 
   \| a^*_{\pm}(\uf) \gs \| \leq c_5 \|f_1\|_{\omega} \cdots \|f_N\|_{\omega},\label{gl3-9} 
\end{equation}
with a constant $c_5$ that is independent of $\alpha$ and $\uf$. If $\feps f_l \in
C^{n+1}_0(\R^3 \backslash \{\fn\}, \C^3)$ for $l=1,...,N$, then there
exists a constant $c_n(\uf)$, such that
\begin{equation}
\big\|a^*_{\pm} (\uf) \gs - e^{iH_{\alpha}t} a^*(\uf_t)
e^{-iH_{\alpha}t} \gs \big\| \leq \alpha^{3/2} \frac{c_n(\uf)}{1+|t|^n} 
\label{gl3-21n}
\end{equation}
\end{proposition}

\begin{proof}
Suppose first, that
$\feps f_1,...,\feps f_N \in C_0^{n+1}(\R^3 \backslash \{\fn\}, \C^3)$.
Then 
\[\frac{d}{dt}\Big( e^{itH_{\alpha}} a^*(\uf_t) e^{-itH_{\alpha}} \gs\Big)=
%ie^{itH_{\alpha}} [W,a^*(\uf_t)] e^{-itH_{\alpha}} \gs=
ie^{i(H_{\alpha}-E_{\alpha})t} [W,a^*(\uf_t)] \gs ,\]
and, by Lemma \ref{l3-1},
\[
\pm \int\limits_t^{\pm \infty} \big\| [W,a^*(\uf_s)] \gs\big\| ds \leq
\alpha^{3/2} \frac{c_n(\uf)}{1+|t|^n}.
\]
This estimate first proves existence of $a^*_{\pm}(\uf)$, 
by Cook's argument, and then it implies \eqref{gl3-21n}.
The existence of $a^*_{\pm}(\uf) \gs$ in the case where
$f_j \in L^2_{\omega}(\R^3 \times \{1,2\})$ now follows from
the approximation argument given in \cite{GZ}, Proposition 2.1.
By the Lemmas \ref{la2} and \ref{la1}
\begin{eqnarray*}
\|e^{itH_{\alpha}} a^*(\uf_t) e^{-itH_{\alpha}} \gs \| &\leq&
\|a^*(\uf_t) (H_f+1)^{-\frac{N}{2}} \| \|(H_f+1)^{\frac{N}{2}} \gs \| \\
&\leq& c_5 \|f_1\|_{\omega} \cdots \|f_N\|_{\omega}, \end{eqnarray*}
uniformly in $t \in \R$ and $\alpha \in [0,\tilde{\alpha}]$.
Letting $t \to \pm \infty$ in this estimate, we obtain (\ref{gl3-9}).
\end{proof}

%%%%%%%%%%%%%%%%%%%%%%%%%%%%%%%%%%%%%%%%%%%%%%%%%%%%%%%%%
%%%%%%%%%%%%%%%%%%%%%%%%%%%%%%%%%%%%%%%%%%%%%%%%%%%%%%%%%
\section{Proofs of the main theorems}

\subsection{A reduction formula} \label{sec4}

In this section we first prove Theorem~\ref{thm4-1} below, which is a
generalization of Theorem~\ref{thm1}, the latter corresponding to the
choice $\tau=0$. The generalization to arbitrary $\tau\in\R$ will be needed in Section~\ref{s4-3}. 

\setcounter{equation}{0}
\begin{theorem} \label{thm4-1}
Let $\feps f_1,...,\feps f_N \in C_0^{2} (\R^3 \backslash \{\fn\},\C^3)$. 
Then
\begin{eqnarray*}
    \lefteqn{a^*_{+}(\uf_{\tau})\gs-a^*_{-}(\uf_{\tau})\gs}\\
 &=& i\alpha^{\frac{3}{2}} \int\limits_{-\infty}^{\infty}
e^{-i(H_0-E_0)\tau}2\fp(s)\ph_{el}\otimes[\fA(\fn,s), a^*(\uf)] \Omega\, ds +\cR(\tau,\alpha)
\end{eqnarray*}
where
$\|\cR(\tau,\alpha)\| = \cO(\alpha^{5/2})+\cO(\alpha^3|\tau|)$ as $\alpha\to 0$.
\end{theorem}

\noindent\textbf{Remark.}
Part of the error $\cO(\alpha^{5/2})$ stems from passing to the
dipole-approximation $W^{(1)} \to \Wdip$. Hence its order $5/2=3/2+1$ cannot be improved.

\begin{proof}
Recall that $B_{t}=B(-t)=e^{-itH_0} B e^{itH_0}$. 
To compare the time-evolutions
generated by $H_\alpha$ and $H_0$ we will use that
\begin{equation}\label{eq:ftc}
    e^{i(H_{\alpha}-E_{\alpha})t}B_{t}\gs = B\gs+
    \int_0^{t}e^{i(H_{\alpha}-E_{\alpha})s}
[iW,B_s]\gs\,ds.
\end{equation}
This equation may be iterated because $[iW,B_s]=[iW_{-s},B]_{s}$. From
$$
    a_{\pm}^{*}(\uf)\gs =
    \lim_{t\to\pm\infty}e^{i(H_{\alpha}-E_{\alpha})t}
a^{*}(\uf_t)\gs
$$
and \eqref{eq:ftc} it follows that
\begin{equation*}
    a^*_{+}(\uf)\gs-a^*_{-}(\uf)\gs
    = \int_{-\infty}^{\infty}e^{i(H_{\alpha}-E_{\alpha})s}
    [iW,a^{*}(\uf_{s})] \gs\, ds.
\end{equation*}
Only terms contributing to this integral of order $\alpha^{3/2}$ need
to be kept. Since
$W=\alpha^{\frac{3}{2}} W^{(1)}+\alpha^3 W^{(2)}$, we may drop
$W^{(2)}$, $W^{(1)}-\Wdip$ and restrict the interval of integration to 
$|s|\leq \alpha^{-1}$ by Lemma \ref{l3-1} and 
Proposition \ref{c3-4n}. We obtain
\begin{eqnarray}
   \lefteqn{a^*_{+}(\uf)\gs-a^*_{-}(\uf)\gs}\nn\\
   &=& i\alpha^{3/2}\int_{-\infty}^{\infty}
e^{i(H_{\alpha}-E_{\alpha})s}[\Wdip ,a^*(\uf_{s})]\gs\,ds +
\cO(\alpha^{5/2})\nn\\
   &=& i\alpha^{3/2}\int_{|s|\leq \alpha^{-1}}
e^{i(H_{\alpha}-E_{\alpha})s}[\Wdip_{-s},a^*(\uf)]_{s}\gs\,ds
+\cO(\alpha^{5/2}). \label{4-1-1} 
\end{eqnarray}
Applying now \eqref{eq:ftc} to the integrand in \eqref{4-1-1} and the
time interval $[\tau,s]$, rather than $[0,s]$, we find
\begin{eqnarray}
\lefteqn{\int_{|s|\leq\alpha^{-1}}e^{i(H_{\alpha}-E_{\alpha})s}[W^{\rm
    dip}_{-s},a^*(\uf)]_{s}\gs\, ds}\nn\\
&=&
e^{i(H_{\alpha}-E_{\alpha})\tau}\int_{|s|\leq\alpha^{-1}}
[\Wdip_{-s},a^*(\uf)]_{\tau}\gs ds\label{4-1-2}\\
 &&+\int_{|s|\leq\alpha^{-1}}ds\int_{\tau}^{s} 
e^{i(H_{\alpha}-E_{\alpha})r}[iW,[\Wdip_{r-s},a^{*}(\uf_r)]]\gs\,dr.\nn
\end{eqnarray}
By \eqref{gl3-22} in Lemma~\ref{l3-4}, the norm of the double integral 
is bounded by
\begin{equation}\label{4-1-3}
   \text{const}\int_{|s|\leq\alpha^{-1}}\frac{|\tau|+|s|}{1+|s|^2}
\alpha^{3/2}\,ds = \cO(\alpha^{3/2}|\tau|)+\cO(\alpha^{3/2}\ln(\alpha)).
\end{equation}
In the integral \eqref{4-1-2} we use Lemma \ref{l3-4} to replace
$\gs$ by $\go$ and to extend the integration over all
$s\in\R$. We find that
\begin{eqnarray}
   \lefteqn{\int_{|s|\leq \alpha^{-1}}
[\Wdip_{-s},a^{*}(\uf)]_{\tau}\gs\,ds}\nn\\
   &=& \int_{-\infty}^{\infty}[W^{\rm
     dip}_{\tau-s},a^{*}(\uf_{\tau})]\go\,ds +\cO(\alpha)\nn\\
   &=& \int_{-\infty}^{\infty}e^{-i(H_0-E_0)\tau}
[\Wdip_{-s},a^{*}(\uf)]\go\,ds+\cO(\alpha).\label{4-1-4}
\end{eqnarray}
Equations \eqref{4-1-1}, \eqref{4-1-2}, \eqref{4-1-3} and
\eqref{4-1-4} prove the theorem because 
$e^{-i(H_{\alpha}-E_{\alpha})\tau}a^{*}_{\pm}(\uf)\gs=
a^{*}_{\pm}(\uf_{\tau})\gs$ and because $\go=\ph_{el}\otimes\Omega$.
\end{proof}

Theorem~\ref{thm4-1} in the case $\tau=0$ becomes Theorem~\ref{thm1},
which implies that 
$$
 \left\|\one_{ac}(\hel)\big(a_{-}^{*}(\uf)\Phi_{\alpha}-a_{+}^{*}(\uf)
\Phi_{\alpha}\big)\right\|^2 = \alpha^3 P^{(3)}(\uf) +O(\alpha^4)
$$
where
\begin{eqnarray}
P^{(3)}(\uf):= \Bigg\|\one_{ac}(\hel)\int\limits_{-\infty}^{\infty}
2\fp(s)\ph_{el}\otimes\big[\fA(\fn,s), a^*(\uf)\big]\Omega ds\Bigg\|^2.
\label{formula1}
\end{eqnarray}
We next show that $P^{(3)}(\uf)$ is additive in its one-photon
contributions.

%----------------- additivity of the ionization probability ------------------------------

\begin{proposition} \label{l3.3}
Suppose that $\uf=(f_1,m_1,\ldots,f_n,m_n)\in
\big[L^2(\R^3)\times\N\big]^n$ with $\sprod{f_i}{f_j}=\delta_{ij}$ and
$\feps f_l \in  C_0^{2}(\R^3 \backslash \{\fn\},\C^3)$. Then
$P^{(3)}(\uf) = \sum_{l=1}^n m_l P^{(3)}(f_l)$
with 
\begin{align}\nonumber
   P^{(3)}(f_l) &= \Bigg\| \one_{ac}(\hel)
   \int\limits_{-\infty}^{\infty} 2\fp(s)\varphi_{el}\cdot \sprod{\Omega}{\fA(\fn,s)a^{*}(f_l)\Omega}\,ds\Bigg\|^2\\
   &= \Bigg\| \one_{ac}(\hel)
   \int\limits_{-\infty}^{\infty} \fx(s)\varphi_{el}\cdot \sprod{\Omega}{\fE(\fn,s)a^{*}(f_l)\Omega}\,ds\Bigg\|^2.\label{gl5-10}
\end{align}
\end{proposition}

\begin{proof}
Since $a^*(\uf)$ is a product of creation operators $a^*(f_l)$ and
since $[\fA(\fn,s),a^*(f_l)]=\sprod{\Omega}{\fA(\fn,s)a^{*}(f_l)\Omega}$, a scalar multiple of the identity
operator, we have 
$$
  \big[\fA(\fn,s),a^*(\uf)\big]\Omega = \sum_{l=1}^n
  \sqrt{m_l}\sprod{\Omega}{\fA(\fn,s)a^{*}(f_l)\Omega}a^*(\uf_{(l)})\Omega
$$ 
where $\uf_{(l)}= (f_1,m_1,\ldots,f_l,(m_l-1),\ldots,f_n,m_n)$. The
vectors $a^*(\uf_{(l)})\Omega$ are orthonormal by construction. Hence
by definition of $P^{(3)}(\uf)$ and by the Pythagoras identity,
$$
   P^{(3)}(\uf) = \sum_{l=1}^n P^{(3)}(f_l)  
$$
with $P^{(3)}(f_l)$ given by the first equation in the statement of the proposition.
The second equation in the proposition follows from 
$$
   2\fp(s)\ph_{el}=\frac{d}{ds}\fx(s)\ph_{el}, \qquad \frac{d}{ds}\sprod{\Omega}{\fA(\fn,s)a^{*}(f_l)\Omega}=-\sprod{\Omega}{\fE(\fn,s)a^{*}(f_l)\Omega}
$$
by an integration by parts. The differentiability of $s\mapsto\fx(s)\ph_{el}$
and the expression for its derivative are established in Lemma~\ref{lb1}.
\end{proof}

%%%%%%%%%%%%%%%%%%%%%%%%%%%%%%%%%%%%%%%%%%%%%%%%%%%%%%%%%%%%%%%%%%%%%%%%%%
\subsection{Expansion in generalized eigenfunctions} \label{sec3.3}

In this section we prove Theorem~\ref{thm3} and the stronger statement expressed by the
Equations~\eqref{intro5b} and \eqref{intro6}. The ingredients are Theorem~\ref{thm1}, Proposition~\ref{l3.3}, and 
a set of generalized eigenfunctions
$\ph_{\fq}$ with the properties (i)-(iii) in
the introduction. Concerning the existence of $\ph_{\fq}$, we recall from
\cite{DerezinskyGerard1997a}, Theorem~4.7.1, that our hypotheses on $V$ imply existence and
completeness of a (modified) wave operator $\Omega_{+}$ associated
with $\hel$. Moreover, $(\hel-i)^{-1}\expect{\fx}^{-2}$ is a Hilbert-Schmidt
operator.

\begin{lemma}\label{lm:one-ac}
Suppose that $\ph:\R\to\Hel\cap D(|\fx|^2)$ is such that
$s\mapsto\ph(s)$ and $s\mapsto |\fx|^2\ph(s)$ are continuous and
absolutely integrable with respect to the norm of $\Hel$. Then
$\int_{-\infty}^{\infty}\ph(s)ds\in D(|\fx|^2)$ and
$$
  \Big\|\one_{ac}(\hel)\int_{-\infty}^{\infty}\ph(s)ds\Big\|^2 =
  \int_{\R^3}\left|\int_{-\infty}^{\infty}\sprod{\ph_{\fq}}{\ph(s)}ds\right|^2 d^3q.
$$
\end{lemma}

\begin{proof}
From the existence of the improper Riemann integrals
$\int_{-\infty}^{\infty}\ph(s)ds$ and
$\int_{-\infty}^{\infty}|\fx|^2\ph(s)ds$ and the fact that
multiplication with $|\fx|^2$ is a closed operator, it follows that
$\int_{-s}^s\ph(s)ds\in D(|\fx|^2)$ and that
$$
     |\fx|^2\int_{-\infty}^{\infty}\ph(s)ds = \int_{-\infty}^{\infty}|\fx|^2\ph(s)ds.
$$
This equation and property (i) of $\ph_{\fq}$ imply that
\begin{align*}
   \sprod{\ph_{\fq}}{\int_{-\infty}^{\infty}\ph(s)ds} &=
   \sprod{|\fx|^{-2}\ph_{\fq}}{\int_{-\infty}^{\infty}|\fx|^2\ph(s)ds}\\
   &= \int_{-\infty}^{\infty}\sprod{|\fx|^{-2}\ph_{\fq}}{|\fx|^2\ph(s)}ds
   = \int_{-\infty}^{\infty}\sprod{\ph_{\fq}}{\ph(s)}ds.
\end{align*}
In view of \eqref{gef1}, this proves the assertion.
\end{proof}

%------------------------------------ thm on the eigenfunction expansion  ----------------------------------------

\begin{proposition}\label{thm:ef-expansion}
Suppose that $\feps f\in C_0^{2}(\R^3\backslash\{\fn\})$. Then
\begin{align*}
    P^{(3)}(f)&= \int_{\R^3} d^3q \left|\sprod{\ph_{\fq}}{\fx\ph_{el}}\cdot
 \int_{-\infty}^{\infty}e^{i(\fq^2-E_0)s}\sprod{\Omega}{\fE(\fn,s)a^{*}(f)\Omega}ds\right|^2 \\
    &=4\pi^2\int_{\R^3}d^3q\Bigg|\sprod{\ph_{\fq}}{\fx\ph_{el}}
\cdot\int_{|\fk|=q^2-E_{0}}
    |\fk| \sum_{\lambda=1,2} \overline{G_\fn(\fk,\lambda)}f(\fk,\lambda)
d\sigma(\fk)\Bigg|^2,
\end{align*}
where $d\sigma(\fk)$ denotes the surface measure of the sphere
$|\fk|=q^2-E_{0}$ in $\R^3$.
\end{proposition}

\begin{proof}
We start with the expression \eqref{gl5-10} for $P^{(3)}(f)$ and we shall apply
Lemma~\ref{lm:one-ac} to
\begin{equation}\label{ef1}
    \ph(s) = \fx(s)\ph_{el}\cdot \sprod{\Omega}{\fE(\fn,s)a^{*}(f)\Omega}.
\end{equation}
By Lemma~\ref{lb1}, $\fx(s)\ph_{el} = e^{i(\hel-E_0)s}\fx\ph_{el}$ belongs
to $D(|\fx|^2)$ and $\||\fx|^2\fx(s)\ph_{el}\|\leq C(1+s^2)$. On the other hand 
\begin{eqnarray}\nonumber
 \sprod{\Omega}{\fE(\fn,s)a^{*}(f)\Omega}
    &=& \sum_{\lambda=1,2} \int i\omega(\fk) e^{-i\omega(\fk)
      s}\overline{G_{\fn}(\fk,\lambda)}f(\fk,\lambda)d^3k\nonumber\\
    &=& \int_0^{\infty} d\omega e^{-i\omega s} \int_{|\fk|=\omega}
    i\omega \sum_{\lambda=1,2}\overline{G_{\fn}(\fk,\lambda)}f(\fk,\lambda)d\sigma(\fk)\label{ef2}
\end{eqnarray}
is the Fourier transform of a function from
$C_0^{\infty}(\R_{+})$, and hence rapidly decreasing as
$s\to\infty$. It follows that \eqref{ef1} satisfies the hypotheses of
Lemma~\ref{lm:one-ac}. Hence Lemma~\ref{lm:one-ac} proves the first asserted equation because
$$
    \sprod{\ph_{\fq}}{\fx(s)\ph_{el}} = e^{i(\fq^2-E_0)s}\sprod{\ph_{\fq}}{\fx\ph_{el}}.
$$
The second equation follows from the first one and from \eqref{ef2} by an
application of the Fourier inversion theorem.
\end{proof}

%-------------------------------------------------------------------------

\subsection{Proof of Theorem~\ref{thm2}}

\begin{lemma} \label{lB-1}
If $\uf=(f_1,...,f_n)$ with $\feps f_1,...,\feps f_n \in
C_0^{\infty}(\R^3 \backslash \{\fn\})$ and $F \in C_0^{\infty}((-\infty,0))$,
then
$$
 a^*_+(\uf) F(H_{\alpha}) - a^*(\uf) F(H_0) = \cO(\alpha^{\frac{3}{2}}) 
$$
\end{lemma}

\begin{proof}
Choose $R \in \R$, such that $\supp(f_1),...,\supp(f_n) \in \{|\fk|<R\}$ and then choose $G \in C_0^{\infty}(\R)$ with
$G=1$ on $\supp(F)+[0,nR]$. Then, by the pull through formula for
$a^*(\uf)$ and by \cite{FGS1}, Theorem~4 (iv),
\[
a^*(\uf) F(H_0)=G(H_0) a^*(\uf) F(H_0),\qquad
a^*_+(\uf) F(H_{\alpha})=G(H_{\alpha}) a^*_{+}(\uf) F(H_{\alpha}). 
\]
Using that $F(H_0)-F(H_{\alpha})=\cO(\alpha^{\frac{3}{2}})$, by the
Helffer-Sj\"ostrand functional calculus, that
$\big(a^*(\uf)-a^*_+(\uf)\big)F(H_{\alpha})=\cO(\alpha^{3/2})$, by the
proof of Proposition~\ref{p3-5}, and that 
$G(H_0) a^*(\uf)$, $a^*_+(\uf) F(H_{\alpha})$ are bounded by Lemma
\ref{la2} and \cite{GZ} Proposition 2.1, we find that
\begin{eqnarray*}
\lefteqn{a^*(\uf) F(H_0)-a^*_+(\uf) F(H_{\alpha})=
G(H_0) a^*(\uf) F(H_0)-G(H_{\alpha}) a^*_+(\uf) F(H_{\alpha})} \\
&=&
G(H_0) a^*(\uf) \Big( F(H_0)-F(H_{\alpha}) \Big) +
G(H_0) \Big(a^*(\uf)-a^*_+(\uf)\Big) F(H_{\alpha}) \\
&&+
\Big( G(H_0)-G(H_{\alpha}) \Big) a^*_+(\uf) F(H_{\alpha})=
\cO(\alpha^{\frac{3}{2}})
\end{eqnarray*}
as $\alpha \to 0$.
\end{proof}

Recall from the introduction that $\cH_{+}^{\alpha}$ is the closure
of the span of all vectors of the form
\begin{equation}
a^*_+(\uh) \gs,\qquad \uh=(h_1,...,h_n),\quad \text{where}\quad
h_i,\omega h_i\in L_{\omega}^2(\R^3 \times\{1,2\}),
\end{equation}
and that $\pap$ is the orthogonal projection onto $\hap$.

\begin{proof}[\textbf{Proof of Theorem~\ref{thm2}}] In the first two
  steps of this proof we shall establish \eqref{pap-pp} in the weak operator
  topology. Then we establish norm convergence to conclude the proof.

\medskip\noindent
\underline{Step 1}: Suppose $\hel \varphi= \lambda \varphi$, $n \in \N$ and
$\uf=(f_1,...,f_n)$ with $\feps f_1,..., \feps f_n \in C_0^{\infty}
(\R^3 \backslash \{\fn\})$. Then 
\begin{equation}
\lim_{\alpha \to 0} \pap\Big(\varphi \otimes a^*(\uf) \Omega\Big) =
\varphi \otimes a^*(\uf) \Omega
\end{equation}
and the analog statement holds for $\varphi \otimes \Omega$.

Since $\lambda <0$ there exists $F \in C_0^{\infty}(\R)$ with
$F(\lambda)=1$ and $\mathrm{supp}(F) \subseteq (-\infty,0)$. Moreover
$\pap F(H_{\alpha})= F(H_{\alpha})$ by the hypothesis of
Theorem~\ref{thm2} and because $\Sigma_{\alpha} \geq 0$
for all $\alpha \in \R$. Using, in addition, that
\[a^*_+(\uf) F(H_{\alpha})-a^*(\uf) F(H_0)=\cO(\alpha^{\frac{3}{2}}),\]
which we know from Lemma \ref{lB-1}, we conclude that
\begin{align*}
\pap \Big(\varphi \otimes a^*(\uf) \Omega\Big)&=
\pap a^*(\uf) F(H_0) \varphi \otimes \Omega=
\pap a^*_+(\uf) F(H_{\alpha}) \varphi \otimes \Omega+
\cO(\alpha^{\frac{3}{2}})\nn\\
&=
a^*_+(\uf) F(H_{\alpha}) \varphi \otimes \Omega+ \cO(\alpha^{\frac{3}{2}})=
\varphi \otimes a^*(\uf) \Omega+ \cO(\alpha^{\frac{3}{2}}).
\end{align*}

Step 1 implies that
$$
\lim_{\alpha \to 0} \pap \Phi=\Phi\quad\text{for all}\ \Phi \in \Ran
\one_{pp}(\hel) \otimes \one_{\F}.
$$

\medskip\noindent
\underline{Step 2}: 
$\ds w-\lim_{\alpha \to 0} \pap (\one_c(\hel) \otimes \one_{\F})=0.$ 

Since $\|\pap (\one_c(\hel) \otimes \one_{\F})\| \leq 1$ for all
$\alpha \in \R$ it suffices to show that
\[ \lim_{\alpha \to 0} 
\lkl a^*_+(\uf) \gs , \pap (\one_c(\hel) \otimes \one_{\F}) \varphi \rkl=0 \]
for all $\varphi \in \cH$ and all $\uf=(f_1,...,f_n)$ with 
$\feps f_1,...,\feps f_n \in C_0^{\infty}(\R^3 \backslash \{\fn\})$.
Since $a^*(\uf) \gs \in \Ran \pap$, this follows from
\[a^*_+(\uf) \gs = a^*(\uf) \go + \cO(\alpha^{\frac{3}{2}}), \]
which follows from Lemma~\ref{lB-1} and Lemma~\ref{la6}.

From Step~1 and Step~2 it follows that
\begin{equation}
w-\lim_{\alpha \to 0} \pap = \one_{pp} (\hel) \otimes \one_{\F}. \label{glB-1}
\end{equation}
Since $P_{+}^{\alpha}$ and $\one_{pp} (\hel) \otimes \one_{\F}$ are
orthogonal projectors, we have
\[\|\pap \varphi\|^2 = \lkl \varphi, \pap \varphi \rkl
\stackrel{\alpha\to 0}{\longrightarrow}
\lkl \varphi, \one_{pp} (\hel) \otimes \one_{\F} \varphi \rkl =
\| \one_{pp} (\hel) \otimes \one_{\F} \varphi \|^2. \]
Combined with \eqref{glB-1} this proves the desired strong convergence.
\end{proof}

%%%%%%%%%%%%%%%%%%%%%%%%%%%%%%%%%%%%%%%%%%%%%%%%%%%%%%%%%%%%%%%%%%%%%%%%%%

\section{Space-Time Analysis of the Ionization Process} \label{s4-3}
\setcounter{equation}{0}
The purpose of this section is to connect our result with those
of the previous papers \cite{BKZ,Z}, where expressions for the zeroth
and first non-trivial order of the ionization probability were defined. We transcribe the definitions
from \cite{Z} to our model and prove their equivalence to the
definitions in this paper. Let $F_R:=\one_{\{|\fx| \geq R\}} \otimes \one_{\F}$.
\begin{proposition} \label{pr5-1}
Let $\feps f_1,...,\feps f_N \in C_0^{2}(\R^3 \backslash \{\fn\}, \C^3)$.
Then
\begin{equation}\label{zeroth} 
   \lim_{R \to \infty} \; \limsup_{\alpha \searrow 0} \;
   \sup_{\tau \in \R} \|F_R a^*_{\pm}(\uf_{\tau}) \gs \|^2 = 0 .
\end{equation}
\end{proposition}

\noindent\textbf{Remarks.} The left hand side of Equation~\eqref{zeroth} may be
interpreted as the ionization probability to zeroth order in
$\alpha$ \cite{Z}. Proposition~\ref{pr5-1} should be
compared to Theorem 4.1 in \cite{Z}.

\begin{proof}
As in the proof of Theorem \ref{thm4-1}
\[a^*_{\pm}(\uf_{\tau}) \gs -a^*(\uf_{\tau}) \gs =
i \int\limits_0^{\pm \infty} e^{is(H_{\alpha}-E_{\alpha})}
[W,a^*(\uf_{\tau+s})] \gs ds, \]
where the integral is $\cO(\alpha^{\frac{3}{2}})$ in norm, uniformly
in $\tau$, by Lemma \ref{l3-1}. Hence it remains to show that
\begin{equation} \lim_{R \to \infty} \limsup_{\alpha \searrow 0} \sup_{\tau \in \R}
\|F_R a^*(\uf_{\tau}) \gs \|^2 =0. \label{gl5-2} \end{equation}
To this end, we observe that, according to Lemma \ref{la2},
\begin{eqnarray*} \|F_R a^*(\uf_{\tau}) \gs \|^2 &\leq&
\|a^*(\uf_{\tau})^2 \gs \| \|F_R \gs \| \\
&\leq& C_{2N} \prod_{l=1}^N \|f_l\|_{\omega}^2 \, \|(H_f+1)^N \gs \|
\, \|\,|\fx| \gs \| \frac{1}{R}. \end{eqnarray*}
This proves (\ref{gl5-2}), because
$\ds \limsup_{\alpha \to 0+} \|(H_f+1)^N \gs \|$ and
$\ds \limsup_{\alpha \to 0+} \| \, |\fx | \gs \|$
are finite by Lemma \ref{la1} and by \eqref{glA-10}.
\end{proof}

\begin{theorem} \label{thm5-2}
Let $\feps f_1,...,\feps f_N \in C_0^{2}(\R^3 \backslash \{\fn\}, \C^3)$,
suppose $\sigma_{sc}(\hel)=\emptyset$, and let $\tau(\alpha)=\alpha^{-\beta}$ 
for some $\beta \in(0,\frac{3}{2})$. Then
\begin{equation}
   P^{(3)}(\uf)=\lim_{R \to \infty} \limsup_{\alpha \searrow 0} \alpha^{-3}
\Big\|F_R \Big[a^*_-(\uf_{\tau(\alpha)}) \gs -
a^*(\uf_{\tau(\alpha)}) \gs \Big]\Big\|^2. \label{gl4-4n}
\end{equation}
\end{theorem}

\noindent
\textbf{Remarks.} Equation~\eqref{gl4-4n} is to be compared with
the expression defining $Q^{(2)}(A)$ in Equation~(1.9) from \cite{Z}: if
we set $g=\alpha^{3/2}$ and $\tau(g)=\alpha^{-\beta}$
in that equation, then $Q^{(2)}(A)$ coincides with
the right hand side of \eqref{gl4-4n}.

\begin{proof}
{From} Proposition \ref{p3-5} we know that
\[\big\| a^*_+(\uf_{\tau(\alpha)}) \gs - a^*(\uf_{\tau(\alpha)}) \gs \big\|
\leq C\frac{\alpha^{3/2}}{\tau(\alpha)}=C\alpha^{3/2+\beta}, \]
hence we may replace $a^*(\uf_{\tau(\alpha)}) \gs$
by $a^*_+(\uf_{\tau(\alpha)}) \gs$ for the proof of \eqref{gl4-4n}.
From Theorem~\ref{thm4-1} we know that
\begin{eqnarray} 
\lefteqn{\lim_{R \to \infty} \limsup_{\alpha \searrow 0} \alpha^{-3}
\|F_R \Big[a^*_-(\uf_{\tau(\alpha)}) \gs -
a^*_+(\uf_{\tau(\alpha)}) \gs \Big] \|^2} \nn \\
&=& \lim_{R \to \infty} \limsup_{\alpha \searrow 0} \|F_R
e^{-i\tau(\alpha)(H_0-E_0)} \Psi(\uf) \|^2 \nn \\
&=& \lim_{R \to \infty} \limsup_{\tau \to \infty}
\|F_R e^{-i\tau \hel} \otimes \one_{\F} \Psi(\uf) \|^2 \label{gl5-5}
\end{eqnarray}
where 
\begin{equation}
  \Psi(\uf) := \int\limits_{-\infty}^{\infty} 2\fp(s)\ph_{el}\otimes
  [\fA(\fn,s) , a^*(\uf)] \Omega\, ds= \sum_{l=1}^N \phi_l \otimes \eta_l. \label{gl5-6}
\end{equation}
Explicit expressions for $\phi_l$ and $\eta_l$ may be taken from 
the proof of Proposition~\ref{l3.3}, e.g., $\eta_l=a^{*}(\uf_{(l)})\Omega$, but
they are not needed here. From \eqref{gl5-6} it follows that
\begin{equation*}
F_R e^{-i\tau \hel} \Psi(\uf)=
\sum_{l=1}^N \Big[ \one_{\{|\fx| \geq R\}} e^{-i\tau \hel} \phi_l \Big]
\otimes \eta_l 
\end{equation*}
where
\begin{eqnarray}
 \one_{\|\fx\| \geq R\}} e^{-i\tau \hel} \phi_l &=&
(\one-\one_{\{\|\fx\| < R \}}) e^{-i\tau \hel } \one_{ac}(\hel) \phi_l \nn\\
&& +\one_{\{\|\fx\| \geq R\}} e^{-i\tau \hel} \one_{pp}(\hel) \phi_l. 
\end{eqnarray}
By the RAGE Theorem, see \cite{W}, Satz~12.8,
\begin{eqnarray}
\lim_{R \to \infty} \sup_{\tau\to \infty}
\|\one_{\{|\fx| \geq R \}} e^{i\tau \hel} \one_{pp}(\hel) \phi_l \|&=& 0, \\
\lim_{\tau\to \infty}
\|\one_{\{|\fx| < R \}} e^{i\tau \hel} \one_{ac}(\hel) \phi_l \|&=& 0.
\label{gl5-9}
\end{eqnarray}
From (\ref{gl5-5})-(\ref{gl5-9}) it follows, that
\[\lim_{R \to \infty} \limsup_{\tau \to \infty} \|F_R e^{-i\tau H_0} \Psi(\uf) \|^2=
\|(\one_{ac}(\hel) \otimes \one_{\F}) \Psi(\uf) \|^2 = P^{(3)}(\uf), \]
by Equation~\eqref{formula1}
\end{proof}

%%%%%%%%%%%%%%%%%%%%%%%%%%%%%%%%%%%%%%%%%%%%%%%%%%%%%%%%%
\begin{appendix}
\section{Uniform  estimates}
\setcounter{equation}{0}
%%%%%%%%%%%%%%%%%%%%%%%%%%%%%%%%%%%%%%%%%%%%%%%%%%%%%%%%%
%
%
In this appendix we collect estimates used in the previous
sections. Most of them are well-known for fixed $\alpha$, but this is
not sufficient for us: we need estimates holding uniformly for $\alpha$ in
a neighborhood of $\alpha=0$. This forces us to review some of the
derivations with special attention to the dependences on $\alpha$. 

\begin{lemma} \label{la2}
For every $N\in\N$ there is a finite constant $C_N$ such that for all 
$h_1,...,h_N \in L^2_{\omega}(\R^3\times\{1,2\})$
\begin{eqnarray}
 \|a^*(\uh) (H_f+1)^{-\frac{N}{2}}\| &\leq& C_N \prod_{l=1}^N \|h_l\|_{\omega} \label{gla3a} \\
 \|a^*(h_{1}) \cdots a^*(h_{l-1}) \fA(\alpha \fx)
a^*(h_{l+1}) \cdots a^*(h_{N}) (H_f+1)^{-\frac{N}{2}}\|
&\leq& C_N \prod_{\substack {j=1 \\j \not =l}}^N \|h_j\|_{\omega}
\label{glA-46}
\end{eqnarray}
\end{lemma}

\begin{proof}
Both, \eqref{gla3a} and \eqref{glA-46} follow from Lemma~17 in
\cite{FGS1}. We recall that 
$\fA(\alpha \fx)=a^*(\fG_{\fx})+a(\fG_{\fx})$
and we note that $\sup_{\fx \in \R^3} \|\fG_{\fx}\|_{\omega}<\infty$.
\end{proof}

%-------------------- Lemmas from FGSi -------------------------------

\begin{proposition}[\cite{FGSi}] \label{pa2}
For every $\lambda < e_1$ there exists a constant $\alpha_{\lambda}>0$, such
that for all $n \in \N$
\begin{equation} 
\sup_{\alpha \leq\alpha_{\lambda}}\big\| |\fx|^n \one_{(-\infty,\lambda]}(H_{\alpha}) \big\| < \infty. 
\end{equation}
\end{proposition}

%---------------- ground state an its energy ---------------------------------------

\begin{proposition}\label{pA-4}
There exists an $\tilde{\alpha} >0$, such that:
\begin{enumerate}
\item[a)] For all $\alpha\leq\tilde{\alpha}$
\begin{equation} E_{\alpha}:=\inf \sigma(H_{\alpha}) \end{equation}
is a simple eigenvalue of $H_{\alpha}$. In the following $\gs$ denotes
the unique normalized ground state of
$H_{\alpha}$ whose phase is determined by $\lkl\gs, \go \rkl \geq 0$.
\item[b)] For every $n \in \N$,
\begin{equation} \sup_{\alpha \leq \tilde{\alpha}}\|\,|\fx|^n\gs\|<\infty.
\label{glA-10} \end{equation}
\item[c)] There exists a finite constant $C$, such that for all
$\alpha \leq\tilde{\alpha}$ and all $\fk \in \R^3 \backslash \{\fn\}$
\begin{eqnarray}
\|a(k) \gs \| &\leq & \alpha^{\frac{3}{2}} C
\frac{|\kappa(\fk)|}{\sqrt{|\fk|}}(1+\alpha |\fk|), \label{glA-26}\\
|E_0-E_{\alpha}| & \leq & \alpha^{\frac{3}{2}}C,\label{glA-3}\\
\|\gs-\go\| &\leq& \alpha^{\frac{3}{2}} C. \label{glA-4n} 
\end{eqnarray} 
\item[d)]
For every $n \in \N$,
\begin{eqnarray}
\sup_{\alpha \leq \tilde{\alpha}} \|[H_f^{n-1},H_{\alpha}](H_{\alpha}+i)^{-n+1} \|
& < & \infty, \label{glA-47c} \\
\sup_{\alpha \leq \tilde{\alpha}} \|H_f^n (H_{\alpha}+i)^{-n} \|
& < & \infty, \label{glA-47a} \\
\sup_{\alpha \leq \tilde{\alpha}} \| \hel (H_{\alpha}+i)^{-1} \|
&<& \infty . \label{glA-47n}
\end{eqnarray}
\end{enumerate}
\end{proposition}

\noindent\textbf{Remark.} Boundedness of
$[H_f^{n-1},H](H+i)^{-n}$ and $H_f^n (H+i)^{-n}$
has previously been established in \cite{FGS1},
Lemma~5, for a class of Hamiltonians $H$ that includes
$H_\alpha$. Yet, that results does not imply \eqref{glA-47c} and
\eqref{glA-47a}, and second, its proof is much more complicated than
the proof of \eqref{glA-47c} and
\eqref{glA-47a}, because $H$ in \cite{FGS1} is defined in terms of a
Friedrichs' extension.

\begin{proof}
That $E_{\alpha}=\inf \sigma(H_{\alpha})$ is an eigenvalue of
$H_{\alpha}$, for small $\alpha$, was first shown in \cite{BFS2}. Its
simplicity follows from \eqref{glA-4n}, which hold for \emph{every}
normalized ground state vector $\gs$ that satisfies the phase condition $\lkl\gs, \go \rkl \geq 0$. A
proof of \eqref{glA-4n} may be found, e.g., in \cite{FGSi}, Proposition~19, Steps 4
and 5. A weaker form of \eqref{glA-26} is given in Lemma~20 of
\cite{FGSi}, but the proof there actually shows \eqref{glA-26}. Estimate
\eqref{glA-3} follows from Lemma~22 in \cite{FGSi} by choosing the
infrared cutoff in this lemma larger than the UV-cutoff. Finally, \eqref{glA-10} is a consequence of Proposition
\ref{pa2} and \eqref{glA-3}.

To prove (d) we set $R_0:=(H_0+i)^{-1}$ and $R_\alpha:=(H_\alpha+i)^{-1}$. It is a
simple exercise, using \eqref{gla3a} and the boundedness of
$(H_f+1)^{1/2}\fp R_0$, to show that $\|WR_0\|=\cO(\alpha^{3/2})$ as
$\alpha\to 0$. Hence we may assume that $\sup_{\alpha\leq \tilde\alpha}\|WR_0\|\leq 1/2$ after
making $\tilde\alpha$ smaller, if necessary. It follows that 
\begin{equation}\label{glA-47}
  \|(H_0+i)R_\alpha\| = \|(1+WR_0)^{-1}\|\leq (1-\|WR_0\|)^{-1}\leq 2
\end{equation}
for all $\alpha\leq\tilde\alpha$. Since $\hel R_0$ and $H_f R_0$ are
bounded operators, we have thus proven statement (d) for $n=1$,
\eqref{glA-47c} being trivial in this case. We now proceed by
induction, assuming that \eqref{glA-47c} and \eqref{glA-47a} hold true
for all positive integers smaller or equal to a given $n\geq 1$. To prove
\eqref{glA-47a} for $n$ replaced by $(n+1)$ we use that 
\begin{eqnarray}
   [H_f^n, H_{\alpha}]R_{\alpha}^n &=& \sum_{l=1}^n \binom{n}{l}
   \mathrm{ad}_{H_f}^l(W)H_f^{n-l}R_\alpha^n\nonumber \\
   &=&  \sum_{l=1}^n \binom{n}{l}
   \mathrm{ad}_{H_f}^l(W)R_\alpha \left(H_f^{n-l}R_\alpha^{n-1}-[H_f^{n-l},H_\alpha]R_\alpha^n\right)\label{glA-48}
\end{eqnarray} 
where $\sup_{\alpha\leq \tilde\alpha}\|
\mathrm{ad}_{H_f}^l(W)R_\alpha\|<\infty$ by \eqref{glA-47}, by
explicit formulas for $\mathrm{ad}_{H_f}^l(W)$ and by the arguments
above proving that $\|WR_0\|=\cO(\alpha^{3/2})$. Hence $\sup_{\alpha \leq \tilde{\alpha}} \|[H_f^n,H_\alpha]R_{\alpha}^n\|<\infty$
follows from \eqref{glA-48} and from the induction hypothesis. Statement
\eqref{glA-47a} with $n$ replaced by $n+1$ now follows from
$H_f^{n+1}R_\alpha^{n+1} = (H_f R_\alpha)(H_f^{n}R_\alpha^{n}) -
H_fR_\alpha[H_f^n,H_\alpha]R_\alpha^{n+1}$, from the induction
hypothesis, and from \eqref{glA-47c} with $n$ replaced by $n+1$, which
we have just established.
\end{proof}

\begin{lemma} \label{la1}
For all $l,m \in \N$:
\begin{eqnarray}
\sup_{\alpha \leq \tilde{\alpha}} \|(H_f+1)^m \gs\| & < & \infty, 
\label{glAA-24} \\
\sup_{\alpha \leq \tilde{\alpha}} \|(H_f+1)^m \lkl\fx\rkl^l \gs \|
&<& \infty, \label{glAA-25}\\
\sup_{\alpha \leq \tilde{\alpha}} \|\fp^2 \gs \| &<& \infty,
\label{glAA-26} \\
\sup_{\alpha \leq \tilde{\alpha}} \|(H_f+1)^m \hel\gs \|
&<& \infty, \label{glAA-27}\\
\sup_{\alpha \leq \tilde{\alpha}} \|(H_f+1)^m \lkl \fx \rkl^l \fp \gs \|
&<& \infty.  \label{glAA-28}
\end{eqnarray}
\end{lemma}

\begin{proof}
The statements \eqref{glAA-24} and \eqref{glAA-26} easily follow from
\eqref{glA-47a}, \eqref{glA-47n}, and \eqref{glA-3}, because $\gs=(H_\alpha+i)^{-n}\gs(E_\alpha+i)^n$; note that
$\fp^2(\hel+i)^{-1}$ is bounded by assumption on $V$. To prove \eqref{glAA-25} we use that 
$$
 \sup_{\alpha \leq\tilde{\alpha}} \|\lkl \fx \rkl^l (H_f+1)^m \gs
 \|^2 \leq\sup_{\alpha \leq\tilde{\alpha}}
\|\lkl \fx \rkl^{2l} \gs \| \cdot \|(H_f+1)^{2m} \gs \|
$$
where the right hand side is finite thanks to \eqref{glA-10} and \eqref{glAA-24}.
To prove \eqref{glAA-27} we write
\begin{eqnarray*}
   \hel(H_f+1)^m\gs &=&
   \hel(H_{\alpha}+i)^{-1}(H_{\alpha}+i)(H_f+1)^m\gs\\
   &=& \hel(H_{\alpha}+i)^{-1}\big[H_{\alpha},(H_f+1)^m\big]\gs\\
    && +  \hel(H_{\alpha}+i)^{-1}(H_f+1)^m\gs (E_{\alpha}+i). 
\end{eqnarray*}
The vectors  $[H_{\alpha},(H_f+1)^m]\gs$ and $(H_f+1)^m\gs$, and the operator
$\hel(H_{\alpha}+i)^{-1}$ are bounded, uniformly in $\alpha\leq \tilde{\alpha}$, by
\eqref{glA-47c}, \eqref{glA-47a} and
\eqref{glA-47n}. This proves \eqref{glAA-27}.

The statement \eqref{glAA-28} follows from \eqref{glAA-25} and
\eqref{glAA-26} after moving both $\fp$'s to one side, and both
factors $\lkl \fx \rkl^l$ to the other side of the inner product
$\|(H_f+1)^m \lkl \fx \rkl^l \fp \gs \|^2 = \lkl (H_f+1)^m \lkl \fx \rkl^l \fp \gs,(H_f+1)^m \lkl \fx \rkl^l \fp \gs\rkl$.
\end{proof}

%--------------------------------------------------------------------------
The following lemma improves upon \eqref{glA-4n}.

\begin{lemma} \label{la6}
For each $m \in \N$ there is a finite constant $K_m$, such that for
all $\alpha\leq \tilde{\alpha}$
\begin{equation} \|(\hel+i) (H_f+1)^m (\gs-\go)\|
\leq K_m \alpha^{\frac{3}{2}}. \end{equation}
\end{lemma}

\begin{proof}
Let $\lambda:=(E_0+e_1)/2$. Thanks to \eqref{glA-3} in Proposition~\ref{pA-4}, we may
assume that $\sup_{\alpha \leq \tilde{\alpha}} E_{\alpha}<\lambda$
by making $\tilde{\alpha}$ smaller, if necessary. Pick $g \in
C_0^{\infty}(\R)$ with $\supp g\subset(-\infty,\lambda)$ and
with $g(E_{\alpha})=1$ for all $\alpha\leq \tilde{\alpha}$. On the one
hand, 
\begin{eqnarray*}
    \lefteqn{\|(\hel+i)(H_f+1)^m g(H_0)(\gs-\go)\|}\\ 
  &\leq & \|(\hel+i)(H_f+1)^m g(H_0)\| \|\gs-\go\| = \cO(\alpha^{3/2})
\end{eqnarray*}
by \eqref{glA-4n}. On the other hand, $(1-g(H_0))(\gs-\go)=(g(H_\alpha)-g(H_0))\gs$ by construction of $g$.
Hence it remains to prove that 
\begin{equation}\label{glA-54}
    \|(\hel+i)(H_f+1)^m(g(H_\alpha)-g(H_0))\gs\| = \cO(\alpha^{3/2}).
\end{equation}

To do so, we use the Helffer-Sj\"ostrand functional calculus with
a compactly supported almost analytic extension $\tilde{g}$ of $g$
that satisfies an estimate $|\partial_{\bar{z}}\tilde{g}(z)|\leq
C|y|^2$. Here and henceforth $z=x+iy$ with $x,y\in\R$. It follows that 
\begin{eqnarray}
    \lefteqn{(\hel+i) (H_f+1)^m (g(H_\alpha)-g(H_0))\gs} \nn\\ 
&=& -\frac{1}{\pi}\int_{\R^2}(\hel+i)(H_0-z)^{-1} (H_f+1)^m
 W(H_{\alpha}-z)^{-1}\gs \frac{\partial \tilde{g}}{\partial\bar{z}}\, dxdy \label{glA-55}
\end{eqnarray}
where
\begin{equation}
   (H_f+1)^m W = \sum_{l=0}^m \binom{m}{l}
   \mathrm{ad}_{H_f}^l(W)(H_f+1)^{m-l}
   =: \alpha^{3/2} \tilde{W}(m)(H_f+1)^m.\label{glA-56}
\end{equation}
From the equations $[H_f,a^{*}(G_x)] =a^{*}(\omega G_x)$ and  $[H_f,a(G_x)] =-a(\omega G_x)$
it is clear that the operator $\tilde{W}(m)$, defined by \eqref{glA-56}, is
$H_0$-bounded. Hence we can estimate the norm of \eqref{glA-55} from
above by
\begin{multline}
   \frac{\alpha^{3/2}}{\pi} \|(\hel+i)(H_0+i)^{-1}\|\int \left|
      \frac{\partial \tilde{g}}{\partial \bar{z}}\right
|\left\|\frac{H_0+i}{H_0-z} \right\|\frac{1}{|z-E_{\alpha}|}\,dxdy\\  
\times \|\tilde{W}(m)(H_0+i)^{-1}\| \|(H_0+i)(H_f+1)^m\gs\|. \label{glA-57}
\end{multline} 
The integral is finite by construction of $\tilde{g}$, because
$|z-E_{\alpha}|^{-1}\leq |y|^{-1}$, and because $\|(H_0+i)(H_0-z)^{-1}\|\leq 1
+ (1+|x|)/|y|$ by the spectral theorem. The last factor in \eqref{glA-57} is bounded
uniformly in $\alpha\leq \tilde{\alpha}$ by \eqref{glAA-24} and \eqref{glAA-27} from Lemma~\ref{la1}.
This establishes \eqref{glA-54} and thus concludes the proof of the lemma.
\end{proof}

%-------------------------- properties of x\ph_{el}  --------------------------------------

\begin{lemma} \label{lb1}
Suppose that $V$ satisfies the hypotheses in Section~\ref{ch2}. Then
\begin{itemize}
\item[(i)] $\fx\ph_{el}\in D(\hel)$ and $(\hel-E_0)\fx\ph_{el}=-2\nabla\ph_{el}$.
\item[(ii)] $e^{-i\hel t}\fx\ph_{el}\in D(|\fx|^2)$ and there exists a
  constant $C$ such that for all $t\in\R$,
$$
    \||\fx|^2e^{-i\hel t}\fx\ph_{el}\| \leq C(1+t^2).
$$
\end{itemize} 
\end{lemma}

\begin{proof}
(i) For all $\gamma\in C_0^{\infty}(\R^3)$ we have $\fx\hel\gamma=\hel
\fx\gamma+2\nabla\gamma$ and hence
\begin{align*}
   \sprod{\hel\gamma}{\fx\ph_{el}}
   &= \sprod{\hel \fx\gamma+2\nabla\gamma}{\ph_{el}}\\
   &= \sprod{\gamma}{E_{0}\fx\ph_{el}-2\nabla\ph_{el}}.
\end{align*}
Since $C_0^{\infty}(\R^3)$ is a core of $\hel$, we conclude that
$\fx\ph_{el}\in D(\hel)$ and that
$$
     \hel \fx\ph_{el} = E_0\fx\ph_{el} -2\nabla\ph.
$$

(ii) Let $\psi:= x_i\ph_{el}$ for some $i\in\{1,2,3\}$. We shall only
need that $\psi\in D(|\fx|^2)\cap D(-\Delta)$ which follows from (i). 
By the fundamental theorem of calculus, in a weak sense
\begin{eqnarray}
\lefteqn{e^{it\hel} |\fx|^2 e^{-it\hel}\psi = \fx^2\psi +
  \int\limits_0^t e^{is\hel} [i\hel,|\fx|^2]e^{-is\hel}\psi ds }\nn\\
&=& |\fx|^2\psi +2 \int\limits_0^t e^{is\hel}
(\fx \cdot \fp + \fp \cdot \fx) e^{-is\hel}\psi ds\nn\\
&=& |\fx|^2 \psi + 2t (\fx \cdot \fp + \fp \cdot \fx)\psi
+2\int\limits_0^t ds \int\limits_0^s dr e^{ir\hel}
(4\fp^2-\fx \cdot \nabla V) e^{-ir\hel}\psi \label{glb-3}.
\end{eqnarray}
Here  $\psi\in D(|\fx|^2)\cap D(-\Delta)\subset
D(\fx\cdot\fp+\fp\cdot\fx)$ and $e^{-ir\hel}\psi\in
D(\hel)=D(-\Delta)$ because $\psi\in D(\hel)$ by part (i). Therefore
assertion (ii) follows from \eqref{glb-3} and from the hypotheses on $V$.
\end{proof}

\end{appendix}

%\bibliographystyle{plain}
%\bibliography{qed}

\end{document}